\documentclass[12pt]{article}
\usepackage[OT1]{fontenc}
\usepackage[onehalfspacing]{setspace}
\RequirePackage{amsthm}
\RequirePackage[cmex10]{amsmath}
\usepackage[cmyk,x11names]{xcolor}
\RequirePackage[colorlinks,linkcolor=blue!70!black,citecolor=black,urlcolor=black]{hyperref}
\usepackage[top=1.23 in, bottom=1.23 in, left=1.25in, right=1.25in]{geometry}
\usepackage{epsfig,setspace,dsfont,array,geometry,wasysym,enumerate,sgame,subcaption,comment,graphicx,tikz,pgfplots,slantsc,ifthen,xfrac,natbib,epstopdf,eqnarray,mathtools,xspace,nicefrac}
\usepackage[shortlabels]{enumitem}

\usepackage[normalem]{ulem}
\usepackage[cmtip,all]{xy}
\usepackage[charter]{mathdesign}

\usepackage[capitalise,noabbrev]{cleveref}
\crefname{section}{section}{sections}
\newtheorem{innercustomgeneric}{\customgenericname}
\providecommand{\customgenericname}{}
\newcommand{\newcustomtheorem}[2]{%
	\newenvironment{#1}[1]
	{%
		\renewcommand\customgenericname{#2}%
		\renewcommand\theinnercustomgeneric{##1}%
		\innercustomgeneric
	}
	{\endinnercustomgeneric}
}

\newtheorem{lemma}{Lemma}
\newtheorem{claim}{Claim}
\newtheorem*{claim*}{Claim}
\newtheorem{theorem}{Theorem}

\newtheorem{corollary}{Corollary}

\newtheorem{proposition}{Proposition}
\newtheorem{definition}{Definition}

\newcustomtheorem{customthm}{Theorem}

\theoremstyle{definition}

\newtheorem*{example*}{Example}
\newcommand{\argmax}{\operatornamewithlimits{argmax}}

\usepackage[compact,small,raggedright]{titlesec}
\titlelabel{\thetitle.\quad}

\begin{document}

	\begin{titlepage}
		
		\title{{\bf \Large{Information Design in Cheap Talk}}\thanks{This paper supersedes the project by the first author under the same title. 
		We would like to thank Ricardo Alonso, Ettore Damiano, Francesc Dilme, Fran\c{c}oise Forges, Fr\'{e}d\'{e}ric Koessler, Christoph Kuzmics, Stephan Lauermann, Gilat Levy, Jin Li, Elliot Lipnowski, Francesco Nava, Matthew Mitchell, Xiaosheng Mu, Ronny Razin, Christopher Sandmann, Xianwen Shi,  Dezs\"{o} Szalay, Bal\'{a}zs Szentes, Jianrong Tian and Yimeng Zhang for helpful discussions at different stages of this research project. We would also like to thank seminar and conference participants at CETC, the Econometric Society meetings, the Parisian Seminar of Game Theory, 34th Stony Brook International Conference on Game Theory, Royal Economic Society Conference, University of Graz, and University of Bonn. 
    Lyu acknowledges support from the Deutsche Forschungsgemeinschaft (DFG, German Research Foundation) under Germany’s Excellence Strategy EXC 2126/1-390838866, and CRC TR 224 (Project B02).
    First version of this paper: 11 July 2022. 
		}}
		\vspace{1.2cm}
		\author{
			\begin{minipage}{0.4\textwidth}\centering  
				Qianjun Lyu
				\\ \centering  \small \it University of Bonn
			\end{minipage}  
			\begin{minipage}{0.4\textwidth}\centering 
				Wing Suen
				\\ \centering  \small \it University of Hong Kong
			\end{minipage}  
		}
		
		\date{\vspace{0.8cm} \today}
		\maketitle
		\thispagestyle{empty}
		\vspace{-0.8cm}
		\begin{abstract}
			
			\noindent An uninformed sender publicly commits to an informative experiment about an uncertain state, privately observes its outcome, and sends a cheap-talk message to a receiver. We provide an algorithm valid for arbitrary state-dependent preferences that will determine the sender's optimal experiment 
			and his equilibrium payoff under binary state space. 
			We give sufficient conditions for informative information transmission. These conditions depend more on marginal incentives---how payoffs vary with the state---than on the alignment of sender's and receiver's rankings over actions within a state. The algorithm can be easily modified to study the canonical cheap talk game with a perfectly informed sender. 
			\vspace{0.5cm}
			
			\noindent \textit{Keywords}: marginal incentives, common interest, concave envelope, quasiconcave envelope, double randomization 
			\vspace{0.5cm}
			
			\noindent \textit{JEL Classification}: D82, D83 
		\end{abstract}

	\end{titlepage}

	\newpage 
	
	\section{Introduction}
	
	Starting from \cite{cheaptalk}, there is a large economics literature that studies how a biased sender	can gain from strategic communication with an uninformed receiver.  
	Much of this literature assumes that the sender is endowed with superior expertise.
	In many scenarios, however,
	the sender needs to learn about the payoff-relevant state before communicating with the receiver. For example, news media and think tanks that are biased for or against a political candidate or a government policy often collect information and conduct research in order to influence public opinion.
	Since the public may not have direct access to the data sources, nor the incentive to use time and effort to assess whether the conclusions drawn indeed follow from the original data,
	these conclusions 
	effectively become cheap-talk messages. Similarly, financial institutions often have research departments whose works provide the basis for their portfolio recommendations to clients,
	but whether their investment advice is consistent with the findings of their research is often unverifiable.
	This paper studies optimal information acquisition when the sender cannot commit to communicating the outcome of his investigations in a verifiable way.

Specifically, we consider a strategic communication game where an uninformed sender 
publicly commits to an information structure (to learn about the state) and privately observes the information outcome before sending a cheap-talk message to a receiver, who then takes an action. 
One can interpret this game as a bridge between strategic communication \citep{cheaptalk} and Bayesian persuasion \citep{kamenica2011bayesian}, in the sense that the sender can commit to the information structure but not to truthful reporting.\footnote{
A canonical cheap talk game with a fully informed sender can be considered as a game where the sender has no commitment to both information structure and truthful reporting, because it is without loss of generality to assume that the sender will acquire perfect information in this case.}

We study this game with a binary state space and finite action space, while allowing the sender to have arbitrary state-dependent preferences. This means that his utility can depend on both the state and the receiver's action. 
State-dependent preferences create a tension between acquiring more information and alleviating conflict of interests. The first incentive is straightforward: acquiring more information enables the sender to make better use of the information.
However, more information can also exacerbate the conflict, potentially intensifying the sender's incentive to misreport. 
Hence, when designing the optimal information structure, the sender needs to anticipate the credibility issues in the interim stage (when different information outcomes are realized). 
	
A generic feature of a model with discrete action space is that the receiver is indifferent between multiple actions at certain beliefs, even though the sender may not be indifferent over those actions. 
To enhance credibility, the sender sometimes benefits from the receiver's randomization between his most and least preferred actions. 
\cite{lipnowski2018cheap} exploit this observation to the case where the sender has state-independent preferences (also referred to as transparent motives).\footnote{
\cite{lipnowski2018cheap} originally study a canonical cheap talk setting, where the sender privately observes the true state. Nevertheless, with state-independent preferences, %their analysis is applicable to our model with an additional stage of information acquisition. 
the set of equilibrium outcomes in these two models coincide. }   
They show that, under transparent motives, the receiver's tie-breaking rules are determined by the sender's indifference between reporting different messages. 
This further leads to their characterization that the sender's highest equilibrium payoff is determined by the quasi-concave envelope of his indirect value function.
However, in cases where the sender's preference depends arbitrarily on the state, such tie-breaking rules are not necessary to guarantee incentive compatibility (as the sender's preference can now vary across beliefs). Hence, a single geometric characterization, such as (quasi)-concavification, is insufficient to determine the highest equilibrium payoff for the sender. 

Nevertheless, we show that with a discrete action space, the sender's indirect value function remains sufficient to determine the optimal information structure and his highest equilibrium payoff. In Section \ref{s:algorithm}, we present a finite algorithm to compute the optimal equilibrium outcome 
for the sender by searching the highest probability that the receiver can take the sender-preferred action without violating the sender's incentive constraints. The optimal information structure generally induces two possible posterior beliefs, and the receiver may take pure or mixed strategies at each of these two beliefs. It turns out that for pure strategies, we can restrict the receiver to take the sender-preferred action. For mixed strategies, the mixing probability is determined by the sender's indifference condition (in a specific way). 
Notably, ``double randomization''---where the receiver takes mixed actions at both posterior beliefs---can be a part of the optimal information design. This contrasts with the case when the sender has transparent motives, where ``double randomization'' is never optimal.

To study the sender's incentive compatibility, 
we call $m_S(a)=u_S(a,1)-u_S(a,0)$ the sender's \emph{marginal incentive} for action $a$. It is the difference in his utility of action $a$ between state 1 and state 0. 
Graphically, the sender's marginal incentives are the slopes of his piecewise indirect value function. Those slopes capture the marginal gain or marginal loss when the sender misreports his private beliefs, and therefore are crucial for incentive compatibility. 

Intuitively one might expect that incentive issues are less severe if sender and receiver have ``similar'' preferences. However, the ``similarity'' we demand here is not captured by whether sender and receiver have the same ranking over actions given some particular state; but rather by the alignment of marginal incentives between them. 
We show in Section \ref{s:valuable}  that, when sender and receiver have opposite marginal incentives (i.e., if the receiver's marginal incentives for action $a$ is higher than $a'$, then the sender's marginal incentives for action $a'$ is higher than $a$), no information can be transmitted at all even if sender and receiver have identical ranking over actions in one state. 

In addition, Section \ref{s:valuable} provides some sufficient conditions  to guarantee informative information transmission when sender and receiver have aligned marginal incentives.  
Information transmission is informative for certain prior beliefs if, from the sender's perspective, (i) no action \emph{blocks} all other actions (``block'' means an action $a$ is better than $a'$ at those beliefs where $a'$ is a best response of the receiver); or (ii) no action is \emph{worst} (i.e., worse than all other actions in both states).  We also consider the case where the sender's preferences are \emph{ordinally state-independent} (i.e., his ranking over actions is the same in the two states). 
In this case, informative information transmission can arise if and only if the sender's ranking over actions is not identical to the receiver's ranking in either of the two states.

	In Section \ref{s:commoninterests}, we discuss an application of our framework where sender and receiver have common interests in one of the states. That is,  the sender's optimal action in state 0 is also the receiver's best response at that state. Then the optimal information structure generates a conclusive signal about state 0 under mild conditions. 
    This application best disentangles the tension between acquiring more information and alleviating conflict of interests. 
   Specifically, revealing state 0 (rather than pooling state 0 with state 1) generates more information for the sender and allows him to better use this information. On top of this, an experiment that identifies the common-interest state would align the two parties' interests ex-post, which 
   raises the sender's ex-ante payoff. This leads to an interesting result that, despite the sender and receiver having common interest in state 0, the optimal information structure does not necessarily reveal the true state with probability one when the true state is indeed of common interests.

 As previously noted, our model represents a middle ground between Bayesian persuasion and canonical cheap talk. We hence explore its connections with Bayesian persuasion in Section \ref{newsection} and with cheap talk in Section \ref{s:cheaptalk}.
 Regarding Bayesian persuasion, the optimal experiment in our model can be more or less informative than that under Bayesian persuasion. Section \ref{newsection} identifies a sensible set of preferences where the optimal experiment in our model is strictly more informative than the optimal experiment in Bayesian persuasion for certain prior beliefs.
Then in Section \ref{s:cheaptalk}, we link our model with the canonical cheap talk model where the sender is perfectly informed at the beginning. 
The equilibrium outcomes in the canonical cheap talk are a subset of the equilibrium outcomes in our model 
due to an additional constraint required for incentive compatibility there. Namely, the sender cannot gain by deviating to a more informative experiment than the one he commits to (in our model). 
This constraint allows us to simplify our algorithm further when searching for the sender-optimal equilibrium in the canonical cheap talk game. In general, the two algorithms can yield different solutions, suggesting that acquiring more information may indeed reduce the sender's equilibrium payoff.

	\textbf{\emph{Related literature.}}
	This paper describes a model of Bayesian persuasion with limited commitment, and is especially close to those papers in this literature that relax the commitment assumption at the communication stage.  In \citet{GuoShmaya2018Costly} and  \cite{tanandanh}, the sender cannot commit to reporting
	the true information outcomes but he incurs a cost of
	making incorrect claims.
	\cite{alonsodata} allow the receiver to endogenously design an audit scheme, which in turn affects the sender's cost of misreporting. \citet{lipnowski2018limit} discuss the situation where the sender can misreport the information outcomes at an exogenously given probability. In \cite{daniel}, the receiver can cross-check the sender's reports by privately randomizing over information structures. 
	Regarding communication 
	games
	with strategic information acquisition, \cite{HarryPei2015} discusses a cheap talk game where the sender can acquire costly information that is unobserved by the receiver.
	\cite{private} consider a promotion game where the sender can privately and sequentially acquire signals generated from a binary experiment. \cite{squintani} allow the sender to choose the number of trials, which can be public information or the sender's private information. In 
	the latter two papers, though information cannot be falsified, its interpretation is subject to the sender's disclosure policies. In contrast to these papers, we assume commitment on information structure and relax the commitment at the communication stage in the sense that the sender's messages are pure cheap talk.

	Our paper contributes to the literature on cheap talk with overt information acquisition. 
	\citet{Ivanov} investigates information design followed by cheap talk in a uniform-quadratic environment. He characterizes the optimal interval information structures. \citet{IngaDezso2019} consider a two-dimensional state space and the sender has access to a signal structure with elliptical distribution. In contemporaneous works, \cite{sophie} studies overt information acquisition with posterior-separable cost in a cheap talk model; 
\cite{lou} instead considers a costless environment with quadratic utilities. 
Both papers formulate their models as a linear persuasion problem with incentive constraints and utilize the recent developments on extreme points (\cite{KleinerMoldovanuStrack2021}, \cite{ArieliEtAl2023}) to show the optimality of bi-pooling information structures. They characterize the optimal information structure under the uniform-quadratic setting.  In contrast, we provide a complete characterization of the optimal information structure for arbitrary preferences within a finite model (binary states and finite actions). This cannot be otherwise accommodated with the existing results on linear persuasion. 
	
Our paper is close to \cite{lipnowski2018cheap}.
They study the canonical cheap talk model where the sender has perfect private information ex-ante. They focus on situations where the sender's preferences are state-independent, and find that the highest equilibrium payoff the sender can achieve is the quasiconcave envelope of his indirect value function. 
In contrast, our sender's private information is endogenously determined by the information structure he commits to. In addition, we allow the sender to have arbitrary state-dependent preferences.  It turns out that, when the sender has state-independent preferences, the equilibrium outcomes in these two models are equivalent. Therefore, the solution of our algorithm under state-independent preferences coincides with their characterization of the quasiconcave envelope. 
However, with state-dependent preferences, there is a fundamental difference between our model and the canonical cheap talk model as discussed in Section \ref{s:cheaptalk}.\footnote{
Other related papers are \cite{Lipknowski2020} and \cite{lucas}.
They provide conditions under which the optimal equilibrium outcome under cheap talk is equivalent to Bayesian persuasion.}

\cite{LinLiu} study the credibility of persuasion assuming that the sender's deviation in messages is not detectable if the marginal distribution of messages remains the same. Their sender's incentive constraints arrive at the ex-ante stage, in the sense that the gain from swapping messages in one state cannot outweigh the loss from that in another state.\footnote{
\cite{LinLiu} focus on pure strategy equilibrium where the receiver cannot randomize.} 
However, our sender's incentive constraints arrive at the interim stage after the outcome of the experiment is privately revealed to the sender. 
The incentive constraints in these two papers are not nested.
\cite{Salamanca} studies a mediated communication game 
in which an informed sender sends a cheap talk message to a mediator, who can commit to a reporting rule based on the sender's message.  
The receiver then takes an action based on the mediator's report. 
Interestingly, under binary state space, our solution provides a lower bound to the sender's highest achievable payoff in \cite{Salamanca}. The relationship is ambiguous for larger state space.  We provide a more thorough discussion in Section \ref{s:discussion}.

Lastly, our paper contributes to the literature on algorithmic information design, see \cite{AlgorithmicBayesianPersuasion}. In a recent work, \cite{babichenko2023} discuss the algorithmic study of a canonical cheap talk game in a finite environment (finite states and finite actions). They show that with certain restrictions, e.g., when the cardinality of state space is constant, the computation will end in polynomial time. Instead, we study a different model with an additional layer of information acquisition. 
Unlike the classic algorithmic approach, our algorithm explicitly writes down the solution of the linear program, which turns out to have clean and simple geometric meanings. 
	
	\section{The Model}
	
	A sender ($S$) and a receiver ($R$) initially share a common prior belief about some state $\theta$.  The state space $\Theta=\{0,1\}$ is binary. We use $\mu \in \Delta\Theta$ to represent a probability distribution over the state, where $\mu(\theta)$ stands for the probability of state $\theta$.  The prior distribution about the state is denoted by $\mu_0$.

	There is a finite set $A$ of actions, with $|A| \ge 2$. We use $a$ to represent a typical element of $A$, and use $\alpha \in \Delta A$ to represent a mixed action (i.e., a probability distribution over $A$).  Each player $i \in \{S,R\}$ is an expected utility maximizer, whose utility $u_i(a,\theta)$ generally depends on both the action and the state. We assume no action is strictly dominated for the receiver.

	The game consists of two stages. In the first stage, the sender commits to choosing a Blackwell experiment (a mapping from the state space to probability distributions over signals) 
	and conducts the experiment at zero cost.
	As is standard in the Bayesian persuasion literature, this is equivalent to choosing a distribution of posterior beliefs induced by the experiment.  In other words, the sender commits to a simple random posterior $P \in \Delta(\Delta \Theta)$ such that $\mathbb{E}_P[\mu]=\mu_0$, and $P$ has a finite support.\footnote{
See \cite{denti}. Because we are directly working with 
the random posterior induced by a Blackwell experiment, 
we implicitly assume, without loss of generality, that distinct signals induce different posterior beliefs.} 
	After the sender conducts the experiment, he privately observes the realization of the random posterior $\mu \in \operatorname{supp}(P)$. We use $P(\mu)$ to denote the ex-ante probability that the experiment induces posterior $\mu$ for the sender (given the prior belief $\mu_0$). The information structure chosen by the sender determines the distribution of his private information. 
	
	In the second stage, the sender interacts with the receiver in a game of strategic information transmission. Denote $M$ as a rich finite message space. 
	Given the random posterior $P$, the sender's reporting strategy,
	$\sigma_S: \operatorname{supp}(P) \rightarrow \Delta M$,
	maps the realization of the random posterior
	to a distribution of messages.
	The receiver's decision rule, 
	$\sigma_R:  M \rightarrow \Delta A$,
	maps the sender's message to a distribution of actions. 
	Each player $i$'s expected utility can be written as:
	\begin{equation*}
		U_i(\sigma_S,\sigma_R,P)=\sum_{\mu\in \operatorname{supp} (P),\, \theta\in \Theta,\,  m\in M,\,  a\in A} P(\mu) \mu(\theta) \sigma_S(m|\mu) \sigma_R(a|m) u_i(a,\theta).
	\end{equation*}
	
	In this framework the sender's posterior belief formation is trivial, and the receiver's posterior belief is obtained from 
	$P$ and $\sigma_S$ 
	using Bayes' rule.
	We focus on Perfect Bayesian Equilibrium,
	and call $(\sigma_S,\sigma_R,P)$ an equilibrium strategy profile if  $\sigma_S$ and $\sigma_R$ are mutual best responses given 
	$P$ and the belief system. 
	The sender chooses the random posterior $P$	
	to maximize his expected utility subject to an equilibrium. If there are
	multiple equilibria for a given $P$,
	we let the sender choose the one 
	that gives him the highest expected utility. 
	
	Notice that each player's equilibrium payoff only depends on the joint distribution of the receiver's posterior belief and the action induced. Therefore, for every equilibrium such that the sender conceals information through a mixed reporting strategy, we can find another truth-telling equilibrium  where the sender directly coarsens the experiment in the first place and the equilibrium outcome remains the same.  
This is reminiscent of the revelation principle.

	\begin{lemma}
		\label{lem1}
		It is without loss of generality to focus on truth-telling equilibria and a random posterior with a binary support, i.e.,
		$|\operatorname{supp}(P)|=|\Theta|=2$. 
	\end{lemma}

    The proof is provided in the Appendix.\footnote{
     This result holds for any finite state space. In particular, the support of the optimal random posterior can have at most $|\Theta|$ elements. 
    } 
    Because there are only two states, it is simpler to represent a probability distribution over the state by the probability of state 1. Henceforth, we use $\mu$ to stand for the probability of state 1. In addition, with binary states, a binary random posterior is completely pinned down by its support given a prior belief $\mu_0$.\footnote{ 
For example, if $\operatorname{supp}( P) = \{\mu',\mu''\}$, then the requirement that $P$ is a mean-preserving spread of the prior belief $\mu_0$ implies that 
$\mu'$ and $\mu''$ are induced with probabilities $P(\mu')$ and $1-P(\mu')$, where $P(\mu')= (\mu''-\mu_0)/(\mu''-\mu')$.} 
	Therefore, we sometimes refer to a binary random posterior simply by its support.

    With slight abuse of notation, let 
	\begin{equation*}
		u_i( a,\mu) := \mu u_i(a,1) + (1-\mu)u_i(a,0)
	\end{equation*}
	be player $i$'s expected utility from action $a$ when player $i$ has posterior belief $\mu$. Let  
	\begin{equation*}
		A_R(\mu) := \argmax_{a\in A}\, u_R(a,\mu)
	\end{equation*}
	be the receiver's best-response correspondence, mapping from belief into a non-empty set of actions. 
	We use $v(\mu):= \operatorname{co} \left(u_S(A_R(\mu),\mu)\right)$ to denote the sender's value correspondence 
	given that both the sender and the receiver hold the same posterior belief $\mu$ and the receiver responds optimally to this belief. Finally, let 
	\begin{equation*}
		\overline{v}(\mu) := \max_{a\in A_R(\mu)}\, u_S(a,\mu)
	\end{equation*}
	be sender's value function when both sender and receiver hold the same belief $\mu$ and the receiver takes the sender-preferred action in his best response correspondence. 
	
	Given Lemma \ref{lem1}, the sender's information design problem can be written as: 
	\begin{equation*}
		\max_{P\in \Delta(\Delta \Theta),\; \sigma_R(a|\cdot)\in \Delta A_R(\cdot)}\quad   \sum_{\mu\in \operatorname{supp} P}  P(\mu) \sum_{a\in  A_R(\mu)} \sigma_R(a|\mu) u_S(a,\mu),
	\end{equation*}
	subject to 
	sender's  incentive constraints: for every $\mu,\mu'\in \operatorname{supp}( P)$,
	\begin{equation}\label{IC-sender}
		\sum_{a\in  A_R(\mu)} \sigma_R(a|\mu) u_S(a,\mu)\ge \sum_{a\in  A_R(\mu')} \sigma_R(a|\mu') u_S(a,\mu),
	\end{equation}
	and subject to the requirement that $|\operatorname{supp}( P)|=2$ and $P$ is a mean-preserving spread of $\mu_0$.
	We denote $W^*(\mu_0)$ as the solution value to this program at the prior belief $\mu_0$.   
	
	Figure \ref{vs} give two examples of the sender's value function $\overline{v}$. The left panel refers to the case where the sender has state-dependent preferences (the piecewise 
	slopes of $\overline{v}$ are arbitrary). 
	The right panel refers to the case where the sender has state-independent preferences ($\overline{v}$ is piecewise constant). The red dashed curves $W^*$ represent the highest equilibrium payoff the sender can achieve for each prior belief (we will elaborate the algorithm to determine $W^*$ in the next section). 
	The function $W^*$ is piecewise affine.

	\begin{figure}[t]
		\centering
		\includegraphics[width=16cm]{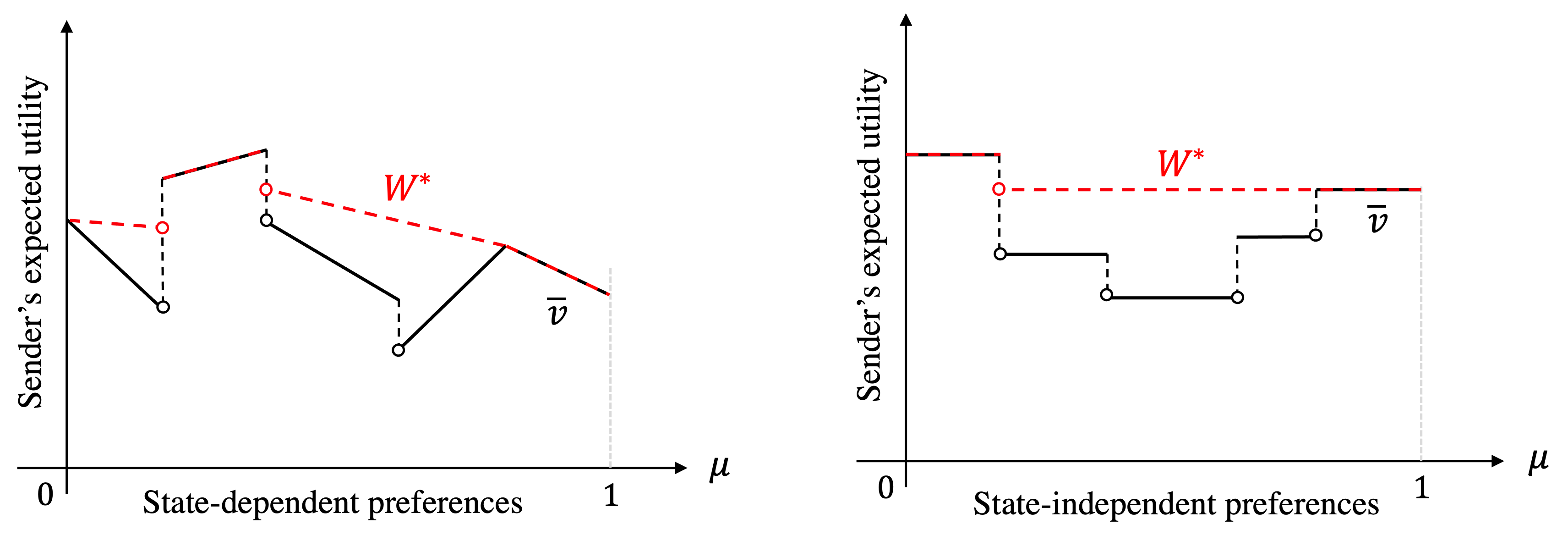}
		\caption{The sender's value function and his highest achievable payoff.}
		\label{vs}
	\end{figure} 
	
	If the sender with arbitrary preferences has full commitment power to truthfully report the outcome of the experiment,
	then the concave envelope of $\overline{v}$ determines the highest equilibrium payoff the sender can achieve 
	\citep{kamenica2011bayesian}. If the sender with state-independent preferences has no commitment power to truthful reporting, the quasiconcave envelope of $\overline{v}$ determines the highest equilibrium payoff the sender can achieve \citep{lipnowski2018cheap}. 
In our model, the sender has arbitrary preferences and no commitment power.  Therefore, $W^*(\cdot)$ is bounded above by the concave envelope of $\overline{v}(\cdot)$.  
	The relationship between $W^*(\cdot)$ and the quasiconcave envelope of $\overline{v}(\cdot)$ is in general ambiguous 
	(see the red curve in the left panel). 
	We will elaborate more on this point later.

	\section{Optimal Information Design}
	\label{s:algorithm}

	We make an assumption about $ A$ in order to clarify the exposition while avoiding burdensome notation.
	We assume that every element in $ A$ is uniquely optimal for the receiver at some belief.
	This rules out the possibility that an action $a\in  A$ is an exact duplicate of another action $a' \in A$ according to the receiver's preferences (i.e., $u_R(a,\theta)=u_R(a',\theta)$ for all $\theta$).
	It also rules out the possibility that $a \in  A$ is weakly optimal (together with $a', a'' \in  A$) for the receiver at exactly one belief, but is strictly worse than $a'$ or $a''$ at any other belief. 
	The analysis in this paper can be suitably extended to handle situations when this assumption does not hold but at the cost of more clumsy notation.

	Given the assumption that every element of $ A$ is 
	a unique best response for the receiver at some belief, we have $|A_R(\mu)| \le 2$ for all $\mu \in [0,1]$.  Moreover, we can order the actions in $ A$ in an increasing sequence, $\{a_{-J},\ldots,a_{-1},a_0,a_1,\ldots,a_K\}$, such that action $a_n$ is
	receiver's best response on a closed interval of beliefs $I_n$, where the lowest belief in $I_n$ is equal to the highest belief in $I_{n-1}$.\footnote{
Specifically, $I_n:=\{\mu\in[0,1]: a_n\in A_R(\mu)\}.$}  
	Here, we let $a_0=A_R(\mu_0)$ be the default action of the receiver when she has no information.  For actions higher than $a_0$, we 
	use $\mu_k$ to denote the \emph{lowest} belief that $a_k$ is a best response for the receiver. For actions lower than $a_0$, we use $\mu_{-j}$ to denote the \emph{highest} belief that $a_{-j}$ is a best response for the receiver. 
	For completeness, we let $\mu_{K+1}=1$ and $\mu_{-J-1}=0$. 
	We call $B:=\{\mu_{-J-1},\ldots,\mu_{-1},\mu_1,\ldots,\mu_{K+1}\}$ the set of \emph{boundary beliefs}.  
	The notation adopted under this convention is illustrated by Figure \ref{vs2}.  Elements of $B$ are highlighted in red.  We assume the prior $\mu_0$ is in the interior of $I_0$ in the figure, but this is not important for our analysis.
	
	\begin{figure}[t]
		\centering
		\includegraphics[width=8cm]{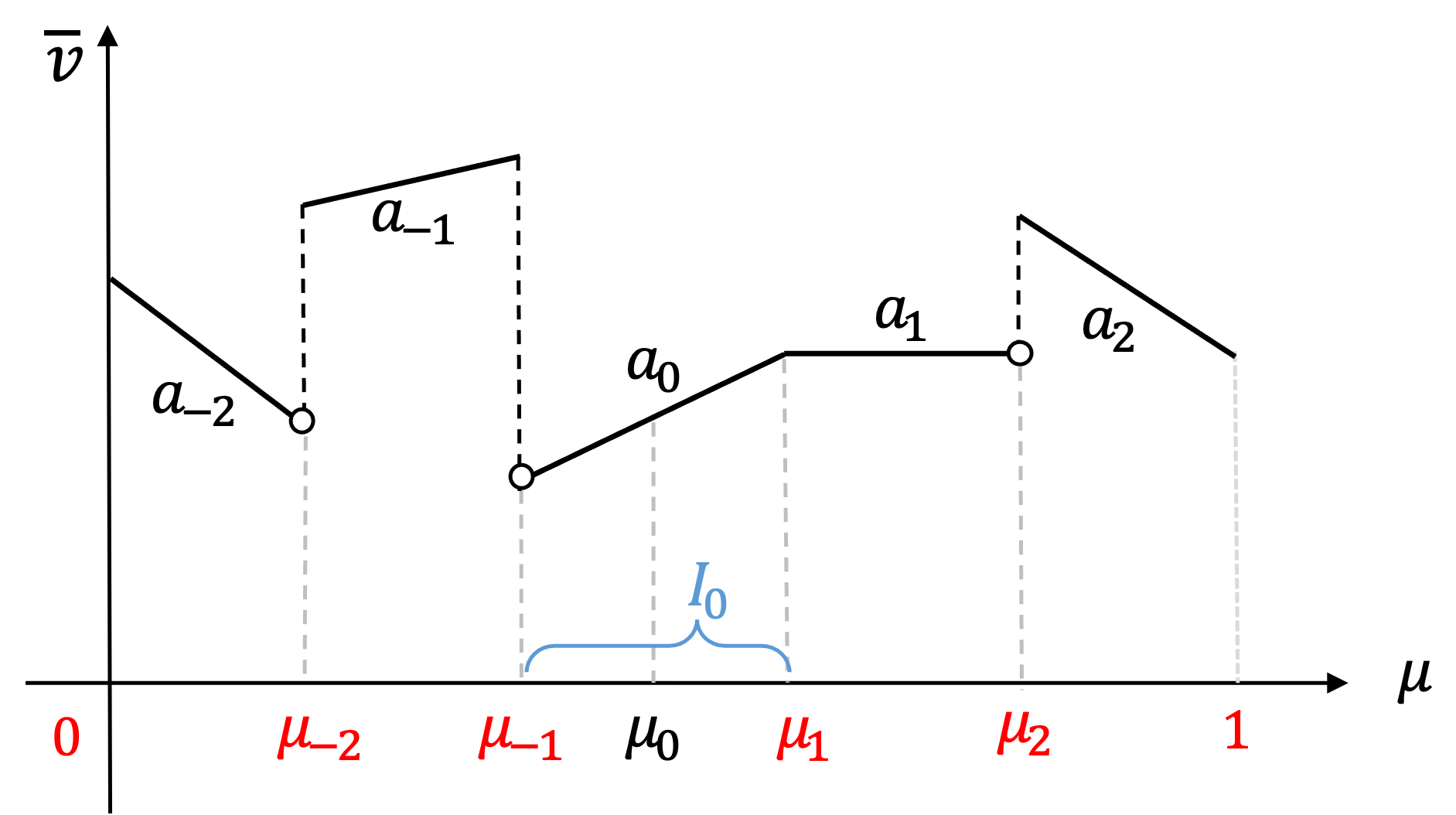}
		\caption{The set of boundary beliefs.}
		\label{vs2}
	\end{figure} 
	
	\begin{proposition}
		\label{withoutloss}
		For any prior belief, there exists an optimal binary random posterior whose support is a subset of the set of boundary beliefs.
	\end{proposition}
	
	\begin{proof}
		If $W^*(\mu_0)=\overline{v}(\mu_0)$, the random posterior with support $\{\mu_{-1},\mu_1\}$ (which induces the default action $a_0$) is optimal.  Suppose $W^*(\mu_0)> \overline{v}(\mu_0)$. Then there is an incentive compatible binary random posterior $P$ with $\operatorname{supp}( P) = \{\mu',\mu''\}$ that induces the receiver to take different responses after different messages, 
  i.e., $\sigma_R(a|\mu')\ne \sigma_R(a|\mu'')$. 
  Suppose that at least one element of $\operatorname{supp}( P)$ does not belong to $B$, say $\mu''\in (\mu_k,\mu_{k+1})$.  Then, the receiver takes pure action $ \sigma_R(a|\mu'')=a_k$ after sender's message $\mu''$. Consider another random posterior $P'$ with $ \operatorname{supp}(P')=\{\mu',\mu_{k+1}\}$, which is strictly more informative than $P$.  
Conditional on that the receiver takes the same action $a_k$ at belief $\mu_{k+1}$, the sender's incentive compatibility constraint (\ref{IC-sender})  at the interim belief $\mu_{k+1}$ still holds given it holds at $\mu''$ (as his expected utility in action $u_S(a,\cdot)$ is linear in belief). 
  Furthermore, incentive compatibility implies that the sender's payoff is convex in belief 
conditional on the receiver's equilibrium actions (i.e., $\max\{u_S(\sigma_R(a|\mu'),\mu),u_S(a_{k+1},\mu)\}$ is convex in $\mu$). 
Because $P'$ is more informative than $P$, his payoff is higher under $P'$ (implied by Blackwell's theorem). A similar reasoning applies when $\mu'$ does not belong to $B$.
	\end{proof}

Proposition \ref{withoutloss} suggests that a finite set of random posteriors is sufficient for determining the optimal information structure.
It is driven by the observation that, for a given pair of actions, if a less informative information structure is incentive compatible,
	then the two parties' interests are aligned for each information outcome, which further implies that a more informative information structure is also incentive compatible and provides the sender with a higher expected utility conditional on that the more informative information structure  induces the same pair of actions on path. Therefore, it is without loss of generality to consider the most informative information structure that can induce a given pair of actions. Every posterior belief induced by 
	this information structure belongs to the set $B$. Henceforth, 
	we can focus on binary random posterior $P$ such that $\operatorname{supp}( P)=\{\mu_{-j},\mu_k\}$ for some $j$ and $k$.

	For a binary random posterior $\{\mu_{-j},\mu_k\}$, use $\alpha_{-j}\in \Delta A_R(\mu_{-j})$ and $\alpha_{k}\in \Delta A_R(\mu_{k})$ to represent the mixed strategy 
	taken after 
	message $\mu_{-j}$ and $\mu_k$, respectively.  
	Let
	\begin{equation*}
		\mathbb{E}_{\alpha_{k}}\left[u_S(a,\mu_{k})\right]
		=
		\sum_{a\in A_R(\mu_{k})}\alpha_{k}(a)\ u_S(a,\mu_{k}) 
	\end{equation*}
	be the sender's expected utility if he has a posterior belief $\mu_k$ and the receiver takes the mixed strategy $\alpha_k$,
	where $\alpha_k(a)$ stands for the probability of taking action $a$ under the mixed strategy $\alpha_k$. Define
	$\mathbb{E}_{\alpha_{-j}}[u_S(a,\mu_{-j})]$ 
	similarly.
	
	Starting with an initial belief $\mu \in (\mu_{-j},\mu_k)$ (i.e., the expectation of the random posterior), the payoff from an experiment that generates posteriors $\mu_{-j}$ and $\mu_k$  and induces $\alpha_{-j}$ and $\alpha_{k}$ is:
	\begin{equation*}
		W_{-j,k}(\mu; \alpha_{-j},\alpha_{k}) :=
		\frac{\mu_k - \mu}{\mu_k-\mu_{-j}} \mathbb{E}_{\alpha_{-j}}[u_S(a,\mu_{-j})] + \frac{\mu-\mu_{-j}}{\mu_k-\mu_{-j}} \mathbb{E}_{\alpha_{k}}[u_S(a,\mu_{k})].
	\end{equation*}
	This payoff is linear in $\mu$ with a constant derivative,
	\begin{equation*}
		W'_{-j,k}(\cdot\ ; \alpha_{-j},\alpha_{k}) = \frac{\mathbb{E}_{\alpha_{k}}[u_S(a,\mu_{k})]-\mathbb{E}_{\alpha_{-j}}[u_S(a,\mu_{-j})]}{\mu_k-\mu_{-j}}.
	\end{equation*}
We provide a geometric illustration of this term in Figure \ref{pp}.
 
 If $\alpha$ puts probability one on an action $a\in A_R(\mu)$, then it represents a pure strategy. We sometimes replace $\alpha$ by $a$ to emphasize the difference between a pure strategy and a mixed strategy.
	
To analyze incentive compatibility, we define the sender's \emph{marginal incentive} corresponding to a mixed strategy $\alpha$ as:
	\begin{equation*}
		m_S(\alpha) := \mathbb{E}_{\alpha}[u'_S(a,\cdot)].
	\end{equation*}
	We also use $m_S(a)=u_S(a,1)-u_S(a,0)$ to represent the marginal incentive for a pure action $a$. 
The sender's marginal incentives for pure actions are the slopes of his piecewise indirect value function. Hence, his marginal incentives for mixed actions are the weighted average of the slopes for the pure actions taken with positive probabilities. To generate credibility for a binary random posterior $\{\mu_{-j},\mu_{k}\}$, both the slopes of $\mathbb{E}_{\alpha_{-j}}[u_S(a,\cdot)]$ and $\mathbb{E}_{\alpha_{k}}[u_S(a,\cdot)]$ and their values at $\mu_{-j}$ and $\mu_{k}$ matter. 
The next 
lemma provides a straightforward method for verifying incentive compatibility.

	\begin{lemma}
		\label{l:IC}
		An information structure that generates posterior beliefs in $\{\mu_{-j},\mu_k\}$ and induces  $\alpha_{-j}$ and $\alpha_{k}$ at these two beliefs satisfies sender's incentive compatibility constraints (\ref{IC-sender}) if and only if 
		\begin{equation}\tag{IC}
			m_S(\alpha_{-j}) \le W'_{-j,k}(\cdot\ ;\alpha_{-j},\alpha_{k}) \le m_S(\alpha_{k}).
			\label{compareslopes}
		\end{equation}
	\end{lemma}
	
	\begin{proof}
		Sender's payoff from inducing $\alpha_{-j}$ at belief $\mu_{k}$ is $\mathbb{E}_{\alpha_{-j}}[u_S(a,\mu_{-j})] + m_S(\alpha_{-j})(\mu_k-\mu_{-j})$.  Incentive compatibility requires that this payoff be lower than $\mathbb{E}_{\alpha_{k}}[u_S(a,\mu_{k})]$, which is sender's payoff from inducing $\alpha_{k}$ at belief $\mu_{k}$.  This is equivalent to $m_S(\alpha_{-j}) \le W'_{-j,k}(\cdot\ ;\alpha_{-j},\alpha_{k})$.  The second inequality in (\ref{compareslopes}) follows similarly from the requirement that the sender has no incentive to induce $\alpha_{k}$ when his private belief is $\mu_{-j}$.
	\end{proof}
	
	Lemma \ref{l:IC} suggests a way to find the optimal information structure.  For each binary random posterior $\{\mu_{-j},\mu_k \}$, we first check condition (\ref{compareslopes}) for all pairs $(\alpha_{-j},\alpha_{k}) \in \Delta A_R(\mu_{-j})\times \Delta A_R(\mu_{k})$, and select the incentive compatible 
 pair with the highest value of 
$W_{-j,k}(\mu_0;\alpha_{-j},\alpha_{k})$.  Optimizing over $j$ and $k$ would then give the highest achievable payoff $W^*(\mu_0)$ for the sender.  

The difficulty is that there are infinitely many pairs $(\alpha_{-j},\alpha_{k})$. 
We identify the most relevant pairs that will guarantee  a solution by searching over such pairs. 
For a random posterior $P$ with support $\{ \mu_{-j},\mu_k\}$, there are three types of receiver's best response we need to consider.
	
	\textbf{\emph{Pure strategy (PP).}}  
	Suppose the receiver takes a pure action after each message. There are four possible PP pairs because the receiver's best response at each boundary belief typically contains two elements. 
Only one pair of actions is necessary for the search. 
Let $\overline{a}_{-j}$ be the sender-preferred action in $A_R(\mu_{-j})$ at belief $\mu_{-j}$, and $\underline{a}_{-j}$ be the remaining action (less preferred by the sender) in $A_R(\mu_{-j})$.  Similarly, let $\overline{a}_{k}$ be the sender-preferred action in $A_R(\mu_k)$ at belief $\mu_{k}$, and $\underline{a}_k$ be the remaining action in $A_R(\mu_k)$. If the sender is indifferent between $A_R(\mu_{-j})$ at belief $\mu_{-j}$, then we let $\overline{a}_{-j}=a_{-j+1}$; and if the sender is indifferent between $A_R(\mu_{k})$ at belief $\mu_{k}$, we choose $\overline{a}_{k}=a_{k-1}$.\footnote{
We break indifference in this way because then the random posterior with support $\{\mu_{-j},\mu_{k}\}$ is the most informative information structure that can induce $a_{-j+1}$ and $a_{k-1}$ if the sender reports truthfully.}
	
	If inequality (\ref{compareslopes}) holds for $(\alpha_{-j},\alpha_{k})=(\overline{a}_{-j},\overline{a}_k)$, we say that the random posterior $P$ is ``IC-PP,'' and we define $W_{-j,k}^{PP} := W_{-j,k}(\mu_0;\overline{a}_{-j},\overline{a}_k)$. 

The random posterior 
$P$ with support $\{\mu_{-j},\mu_k\}$ in Figure \ref{pp} is 
IC-PP. 
To see this, note that the inequalities in (\ref{compareslopes}) are geometrically equivalent to the following: the slope of the left black piece $m_S(\overline{a}_{-j})$ is smaller than the slope of the middle orange piece $W'_{-j,k}(\cdot\ ;\overline{a}_{-j},\overline{a}_{k})$, which is smaller than the slope of the right black piece $m_S(\overline{a}_{k})$. By this sequence of inequalities, when 
$u_S(\overline{a}_{-j},\cdot)$ (the left black piece) is extended to $\mu_k$, its value is below $u_S(\overline{a}_k,\mu_k)$ (the black dot on the right).  This indicates that the sender would not misreport $\mu_{-j}$ when his true belief is $\mu_{k}$.  Similarly, he has no incentive to misreport $\mu_k$ when his true belief is $\mu_{-j}$.

	\begin{figure}[t]
		\centering
		\includegraphics[scale=0.55]{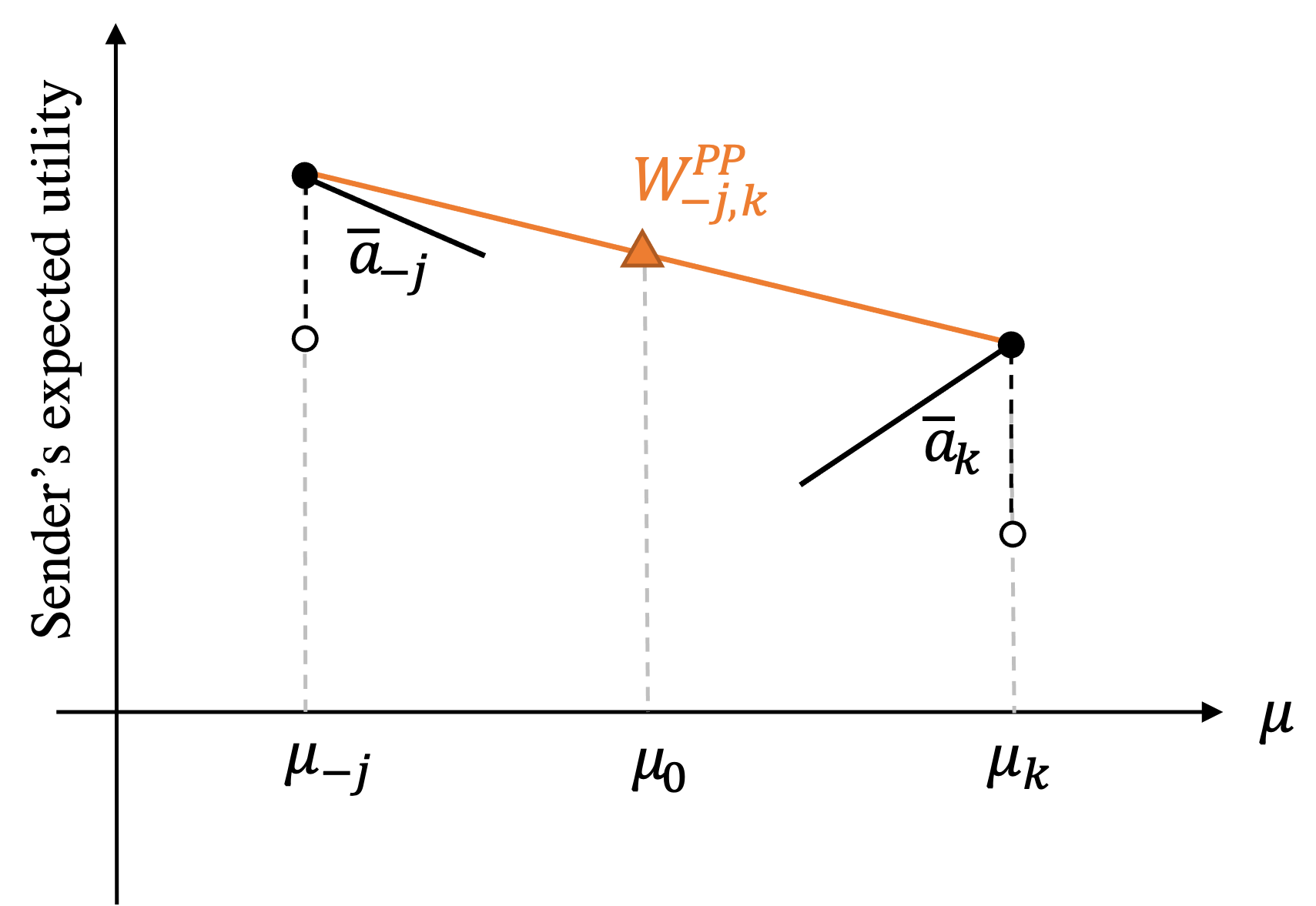}
		\caption{Incentive compatibility for pure strategy.}
		\label{pp}
	\end{figure}

	\textbf{\emph{One-sided randomization (PM or MP).}}  
	Suppose the receiver takes a mixed strategy after one of the messages.  Consider the case of PM (the MP case is symmetric),  
	and consider the pair $(\alpha_{-j},\alpha_k)=(\overline{a}_{-j},\alpha_k^{PM})$, where $\alpha_k^{PM}$ 
	puts weight $\gamma_k$ on $\overline{a}_k$ and weight $1-\gamma_k$ on $\underline{a}_{k}$.  The value of $\gamma_k$ is determined by the requirement that, at belief $\mu_{-j}$, the sender is indifferent between his preferred action $\overline{a}_{-j}$ and the mixed action $\alpha_k^{PM}$, 
	\begin{equation}
	\label{pmindi}
		u_S(\overline{a}_{-j},\mu_{-j})=\gamma_k u_S(\overline{a}_k,\mu_{-j})+(1-\gamma_k)u_S(\underline{a}_{k},\mu_{-j}).
	\end{equation}
	Notice that the indifference condition (\ref{pmindi}) determines the highest probability that the receiver can take the sender-preferred action $\overline{a}_{k}$ at belief $\mu_k$ without violating the sender's incentive compatibility at belief $\mu_{-j}$. By construction, the pair $(\overline{a}_{-j},\alpha_k^{PM})$ satisfies 
\begin{equation*}
W'_{-j,k}(\cdot\ ;\overline{a}_{-j},\alpha_k^{PM})=m_S(\alpha_k^{PM}), 
\end{equation*}
i.e., the second inequality in (\ref{compareslopes}) holds with equality.  
If it also satisfies the first inequality in (\ref{compareslopes}), and if $\alpha_k^{PM}$ is a valid mixed action,\footnote{
The value of $\gamma_k$ that satisfies equation (\ref{pmindi}) may be outside $[0,1]$, in which case $\alpha_k^{PM}$ is not a probability distribution.  } 
we say that the random posterior $P$ is ``IC-PM,'' and we define $W_{-j,k}^{PM} := W_{-j,k}(\mu_0;\overline{a}_{-j},\alpha_k^{PM})$.\footnote{
If  $\alpha_k^{PM}$ is not a valid probability distribution,  we let $\mathbb{E}_{\alpha_{k}}\left[u_S(a,\mu_{k})\right]:=\gamma_k u_S(\overline{a}_k,\mu_{-j})+(1-\gamma_k)u_S(\underline{a}_{k},\mu_{-j})$, given that $\gamma_k$ is the solution to equation (\ref{pmindi}). The corresponding value of $W_{-j,k}^{MP}$ is defined accordingly.  We adopt a similar convention for the cases of 
PM and MM.}

	The left panel of Figure \ref{pm} illustrates 
the geometric construction of one-sided randomization. 
	First, draw an affine line connecting the left black dot $u_S(\overline{a}_{-j},\mu_{-j})$ and the green dot on the right $($that is the intersection point between the extended curves of  $u_S(\overline{a}_{k},\cdot)$ and $u_S( \underline{a}_{k},\cdot))$. If this affine line (colored in orange) intersects with the sender's value correspondence $v(\cdot)$ at belief $\mu_k$, then the blue dot at that intersection represents the mixed action $\alpha_k^{PM}$. In addition, we observe that the slope of the left black piece is smaller than the slope of the orange affine piece. Therefore, the value of this orange affine line at the prior belief $\mu_0$ is the sender's expected utility from the one-sided randomization that we identify.

	\begin{figure}[t]
		\centering
		\includegraphics[scale=0.555]{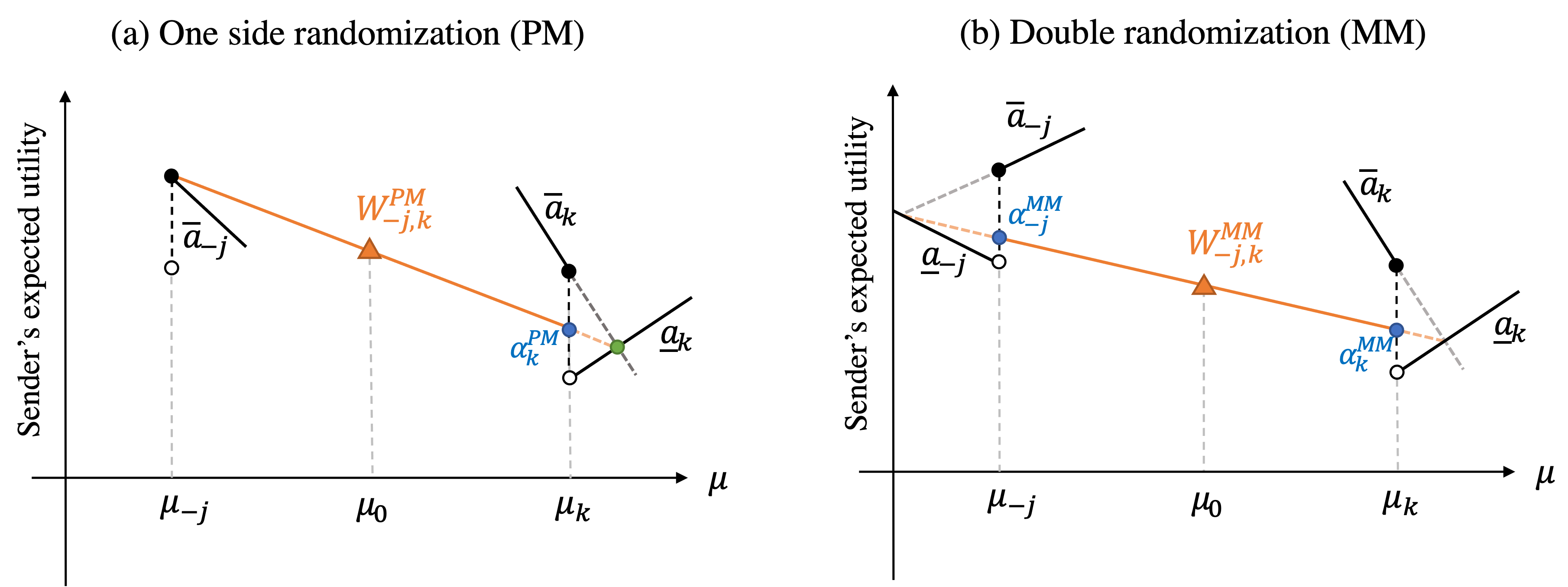}
		\caption{Relaxing incentive constraints by randomization.}
		\label{pm}
	\end{figure}

	\textbf{\emph{Double randomization (MM).}}  
	This involves the receiver taking a mixed strategy after each message.  Let $\alpha_{-j}^{MM}$ be a mixed action that puts weight $\gamma_{-j}$ on $\overline{a}_{-j}$ and weight $1-\gamma_{-j}$ on $\underline{a}_{-j}$.  Let $\alpha_{k}^{MM}$ be a mixed action that puts weight $\gamma_k$ on $\overline{a}_k$ and weight $1-\gamma_k$ on $\underline{a}_{k}$.  The weights $\gamma_{-j}$ and $\gamma_k$ are chosen in such way that the sender is indifferent between  $\alpha_{-j}^{MM}$ and $\alpha_{k}^{MM}$ both at belief $\mu_{-j}$ and at belief $\mu_k$:
	\begin{equation}\label{mmcondition}
		\mathbb{E}_{\alpha_{-j}^{MM}}[u_S(a,\mu_{-j})]=\mathbb{E}_{\alpha_{k}^{MM}}[u_S(a,\mu_{-j})], \qquad
		\mathbb{E}_{\alpha_{k}^{MM}}[u_S(a,\mu_{k})]=\mathbb{E}_{\alpha_{-j}^{MM}}[u_S(a,\mu_{k})].
	\end{equation}
	The two indifference conditions (\ref{mmcondition}) determine the highest probability that the receiver can take the sender-preferred action at both beliefs $\mu_{-j}$ and  $\mu_k$ without violating the sender's incentive compatibility at both beliefs.  
	As by construction, $(\alpha_{-j}^{MM},\alpha_k^{MM})$ satisfies $m_S(\alpha_{-j}^{MM})=W'_{-j,k}(\cdot\ ;\alpha_{-j}^{MM},\alpha_k^{MM})=m_S(\alpha_k^{MM})$, and so the incentive constraints (\ref{compareslopes}) hold. If the value of $(\gamma_{-j},\gamma_k)$ that solves these two equations lies inside $[0,1]^2$, then 
	both $\alpha_{-j}^{MM}$ and $\alpha_{k}^{MM}$ are valid mixed actions.\footnote{
    Whenever $\gamma_j$ and $\gamma_k$ are not uniquely pinned down by the indifference conditions (e.g., when all four actions share the same marginal incentives), we pick the values that maximize the sender's ex ante payoff.}   
    Then we say that the random posterior $P$ is ``IC-MM,'' and we define $W_{-j,k}^{MM}=W_{-j,k}(\mu_0;\alpha_{-j}^{MM},\alpha_k^{MM})$.

The right panel of Figure \ref{pm} illustrates the construction of double randomization. We first draw an orange affine line connecting the intersection points of the two left black pieces and the two right black pieces. The sender's expected utilities $u_S(\alpha_{-j}^{MM},\cdot)$,  $u_S(\alpha_{k}^{MM},\cdot)$, and his expected payoff from such double randomization $W_{-j,k}(\mu_0;\alpha_{-j}^{MM},\alpha_k^{MM})$ all coincide on this (orange) affine line. 
In addition, if this (orange) affine line intersects with the sender's value correspondence $v(\cdot)$ at both $\mu_{-j}$ and $\mu_{k}$ (i.e., the two blue dots in Figure \ref{pm}(b) lies in his value correspondence), then $\alpha^{MM}_{-j}$ and $\alpha^{MM}_{-j}$ are valid mixed actions and the random posterior $\{\mu_{-j},\mu_k\}$ is 
IC-MM.

	Now we introduce an algorithm that yields the highest achievable payoff $W^*(\mu_0)$, together with an implied optimal random posterior $P^*$.

	\textbf{\emph{Algorithm 1:}}
	\setlist{nolistsep}
	\begin{enumerate}
		
		\item For every pair $(-j,k)\in 
		\{-J-1,\ldots,-1\}\times \{1,\ldots,K+1\}$, 
		compute $W_{-j,k}(\mu_0;\overline{a}_{-j},\overline{a}_k)$ and rank these values from highest to lowest.\footnote{
			It is not important how we break ties. }  
		Starting from the pair with the highest value, verify whether it is IC-PP or not.  Stop the first time an IC-PP pair is found. Assign $W^1=W_{-j,k}^{PP}$ for such pair and let the set of $(-j,k)$ pairs with $W_{-j,k}^{PP}$ strictly higher than $W^1$ be $S_1$.  If there does not exist an IC-PP pair, assign $W^1=\overline{v}(\mu_0)$ and let $S_1=
		\{-J-1,\ldots,-1\}\times \{1,\ldots,K+1\}$, 
		\vspace{0.2cm}
		
		\item For every pair $(-j,k)$ in $S_1$:  
		\begin{enumerate}
			\item Compute $W_{-j,k}(\mu_0;\overline{a}_{-j},\alpha_k^{PM})$ and re-rank these values from highest to lowest. 
			Starting with the pair with the highest value, verify whether it is IC-PM or not.  Stop the first time when an IC-PM pair is found.  Assign $W^{(a)}=W_{-j,k}^{PM}$ for such pair and let the set of $(-j,k)$ pairs with $W_{-j,k}^{PM}$ strictly higher than  $W^{(a)}$ be $S^{(a)}$. If none of them is IC-PM, assign $W^{(a)}=\overline{v}(\mu_0)$ and $S^{(a)}=S_1$
			\item Go through a symmetric procedure in the case of MP. Assign $W^{(b)}=W_{-j,k}^{MP}$ the first time an IC-MP pair is found and let the set of $(-j,k)$ pairs with $W_{-j,k}^{MP}$ strictly higher than  $W^{(b)}$ be $S^{(b)}$. If none of them is IC-PM, assign $W^{(b)}=\overline{v}(\mu_0)$ and $S^{(b)}=S_1$. 
			\item Let $W^2=\max\{W^{(a)}, W^{(b)}\}$.  Let $S_2= S^{(a)} \cup S^{(b)}$.
		\end{enumerate}
		\vspace{0.2cm}
		
		\item For every pair $(-j,k)$ in $S_2$, compute $W_{-j,k}(\mu_0;\alpha_{-j}^{MM},\alpha_{k}^{MM})$ and re-rank these values from the highest to lowest. Starting with the pair with the highest value, verify whether it is IC-MM or not. Stop the first time an IC-MM pair is found and assign $W^3=W_{-j,k}^{MM}$ for such pair. If none of them is IC-MM, assign $W^3=\overline{v}(\mu_0)$. 
		\vspace{0.2cm}
		
		\item Assign $W^*(\mu_0)=\max\{W^1,W^2,W^3\}$. The random posterior with support $\{\mu_{-j},\mu_k\}$ corresponding to the $(-j,k)$ pair that yields $W^*(\mu_0)$ is optimal.  
	\end{enumerate}

	\begin{theorem}
		\label{theorem}
		Algorithm 1 determines the highest achievable payoff for the sender. 
	\end{theorem} 
	
   The algorithm above specifies a finite procedure to determine the sender's highest equilibrium payoff. In principle, for every pair $(-j,k)$, it is sufficient to search only four possibilities, namely IC-PP, IC-PM, IC-MP, and IC-MM as they determine the highest probabilities that the receiver can take the sender-preferred action without violating the sender's incentive compatibility. 
   Therefore, the steps of the procedure are bounded above by $|A|^2+4|A|$.
   Nevertheless, the procedure we describe guarantees a faster search without checking all possibilities across all $(-j,k)$. We prove the sufficiency of such simplification in the appendix.

	To find the highest equilibrium payoff across different prior beliefs
	as in Figure \ref{vs}, in principle, we would run Algorithm 1 for every prior belief $\mu_0$. However, it is unnecessary given the linearity of the problem. 
We only need to apply the algorithm again when the prior belief crosses a boundary belief, i.e., 
when the set of $(-j,k)$ satisfying Bayesian plausibility changes. 
	
	The construction behind this algorithm generalizes \cite{lipnowski2018cheap} under the case of a binary state with arbitrary preferences.  When the sender has state-independent preferences (transparent motives), the marginal incentive $m_S(\alpha)$ is equal to 0
	for every mixed action $\alpha$ (including pure action).  The incentive compatibility requirement (\ref{compareslopes}) in Lemma \ref{IC-sender} would then require $W'_{-j,k}(\cdot\ ;\alpha_{-j}, \alpha_k)=0$ for any action pair. 
	This implies that to find the sender's highest achievable payoff, we can search for the highest piecewise step functions such that every endpoint of a piece is inside the sender's value correspondence. This leads to 
	the quasiconcave envelope of $\overline{v}(\cdot)$.  In our setup, the fact that $m_S(\alpha_{-j})$ is in general different from $m_S(\alpha_{k})$ means that $W'_{-j,k}(\cdot\ ;\alpha_{-j}, \alpha_k)$ is not restricted to be equal to 0.  The sender in our setup can achieve a payoff greater than or less than the quasiconcave envelope of $\overline{v}(\cdot)$.

	The use of randomization to relax incentive compatibility constraints is emphasized in \cite{lipnowski2018cheap}. Nevertheless, double randomization is never optimal under transparent motives.  To see this, if the sender is recommending mixed actions $\alpha_{-j}$ and $\alpha_{k}$ at beliefs $\mu_{-j}$ and $\mu_k$, he could strictly raise his payoff by putting more weight on $\overline{a}_{-j}$ and $\overline{a}_{k}$ in these mixed actions, provided that the new pair of mixed actions are still incentive compatible.  Such deviation is always feasible as long as marginal incentives $m_S(\cdot)$ are equal for all actions.
	In our model with general preferences, such deviation may not be feasible, and therefore double randomization can remain a candidate as part of optimal information design.

    \section{Positive Information Transmission}
	\label{s:valuable}

In this section, we explore the existence of informative information transmission in our model. Algorithm 1 in Section \ref{s:algorithm} determines the sender's maximum payoff $W^*(\mu_0)$ under an optimal information structure. We say that there is 
\emph{positive information transmission} 
if $W^*(\mu_0) > \overline{v}(\mu_0)$ for some prior belief $\mu_0 \in [0,1]$. 
Otherwise, 
if $W^*(\mu_0)=\overline{v}(\mu_0)$ for all $\mu_0\in [0,1]$, then no information can be transmitted on path. 
	
 In \cite{kamenica2011bayesian} or \cite{lipnowski2018cheap}, information transmission is positive if and only if the sender's indirect value function $\overline{v}(\cdot)$ is not concave or not quasiconcave, respectively.\footnote{
		In a model with discrete action space, the sender's value function $\overline{v}(\cdot)$ is (generically) discontinuous at beliefs for which the receiver is indifferent between different actions.  Since a discontinuous function is not concave, 
        information transmission is always positive according to our definition when there is full commitment.}
	In contrast, in our model, whether 
 information transmission is positive depends less on the concavity properties of $\overline{v}(\cdot)$ 
 and more on the structure of marginal incentives $m_S(\cdot)$. 
 In other words, there is no straightforward way to characterize the necessary and sufficient condition for positive information transmission. In the following, we provide some economically meaningful sufficient conditions that will settle this question.

	We introduce the following concepts that relate to the conflict of interest between sender and receiver.
	
	\begin{definition}\label{opposite}
		\textnormal{Sender and receiver have \emph{opposite marginal incentives} if, for any $a',a'' \in  A$, 
			\begin{equation*}
				m_R(a') < m_R(a'') \iff m_S(a') > m_S(a'').
			\end{equation*}
			They have \emph{aligned marginal incentives} if, for any $a',a'' \in   A$,}
		\begin{equation*}
			m_R(a') < m_R(a'') \iff m_S(a') < m_S(a'').
		\end{equation*}
	\end{definition}
	
	The notion of opposite or aligned marginal incentives has little to do with comparing the level (or the ranking) of utilities attached to different actions at a given belief by the receiver and by the sender.  For example,  sender and receiver may have identical preference ranking over actions in $ A$ if they know the true state is, say, state 0; yet they may still have opposite marginal incentives according to Definition \ref{opposite}.
	
	Our definition is related to supermodularity or submodularity between action and state.  With a binary state space, it is without loss of generality to assume that the receiver preferences are supermodular in $(a,\theta)$ (because we order actions in such a way that higher actions are chosen at higher beliefs).  According to this convention, if $u_S(\cdot,\cdot)$ is
	strictly submodular, then sender and receiver have opposite marginal incentives.  If $u_S(\cdot,\cdot)$ is 
	strictly supermodular, they have aligned marginal incentives.

\subsection{Opposite marginal incentives}
    
	\begin{proposition}
		\label{oppoprop}
		If sender and receiver have opposite marginal incentives, then no information can be transmitted. 
	\end{proposition}

	\begin{proof}
		Consider an arbitrary prior belief $\mu_0 \in (0,1)$.  Take any pair of boundary beliefs such that $\mu_{-j} < \mu_0 < \mu_k$, and take any arbitrary receiver's best responses $\alpha_{-j} \in \Delta A_R(\mu_{-j})$ and $\alpha_k\in \Delta A_R(\mu_{k})$, with $\alpha_{-j} \ne \alpha_{k}$.  Our convention of ordering actions implies that $m_R(\alpha_{-j}) < m_R(\alpha_{k})$, and hence $m_S(\alpha_{-j}) > m_S(\alpha_{k})$.  By Lemma \ref{IC-sender}, this pair of actions $(\alpha_{-j},\alpha_k)$ cannot be incentive compatible.  This means that there does not exist an incentive compatible (binary) random posterior that can induce different actions at different interim beliefs. In other words, in any equilibrium, the receiver's action does not depend on the sender's messages. That is, no information can ever be transmitted.  
	\end{proof} 
	
	Proposition \ref{oppoprop} is valid regardless 
of whether sender and receiver have the same or different rankings over the set of actions in the two states.
As long as their marginal incentives are opposite, information  cannot be transmitted.\footnote{
This result is related to \cite{LinLiu} in their model of credible persuasion. With opposite marginal incentives, the  sender's overall net gain from swapping messages in both states is always positive. Therefore they cannot generate incentive compatibility in their model. 
} 
Figure \ref{oppoex} shows one such example.  The sender's value function $\overline{v}(\cdot)$ in this figure is obviously not concave.  Nevertheless, because the slope in each separate segment of $\overline{v}(\cdot)$ is decreasing, Proposition \ref{oppoprop} implies that, despite 
the sender's power to commit to an information structure,
this cannot improve his payoff from a babbling equilibrium for any prior belief.

	\begin{figure}[t]
		\centering
		\includegraphics[scale=0.6]{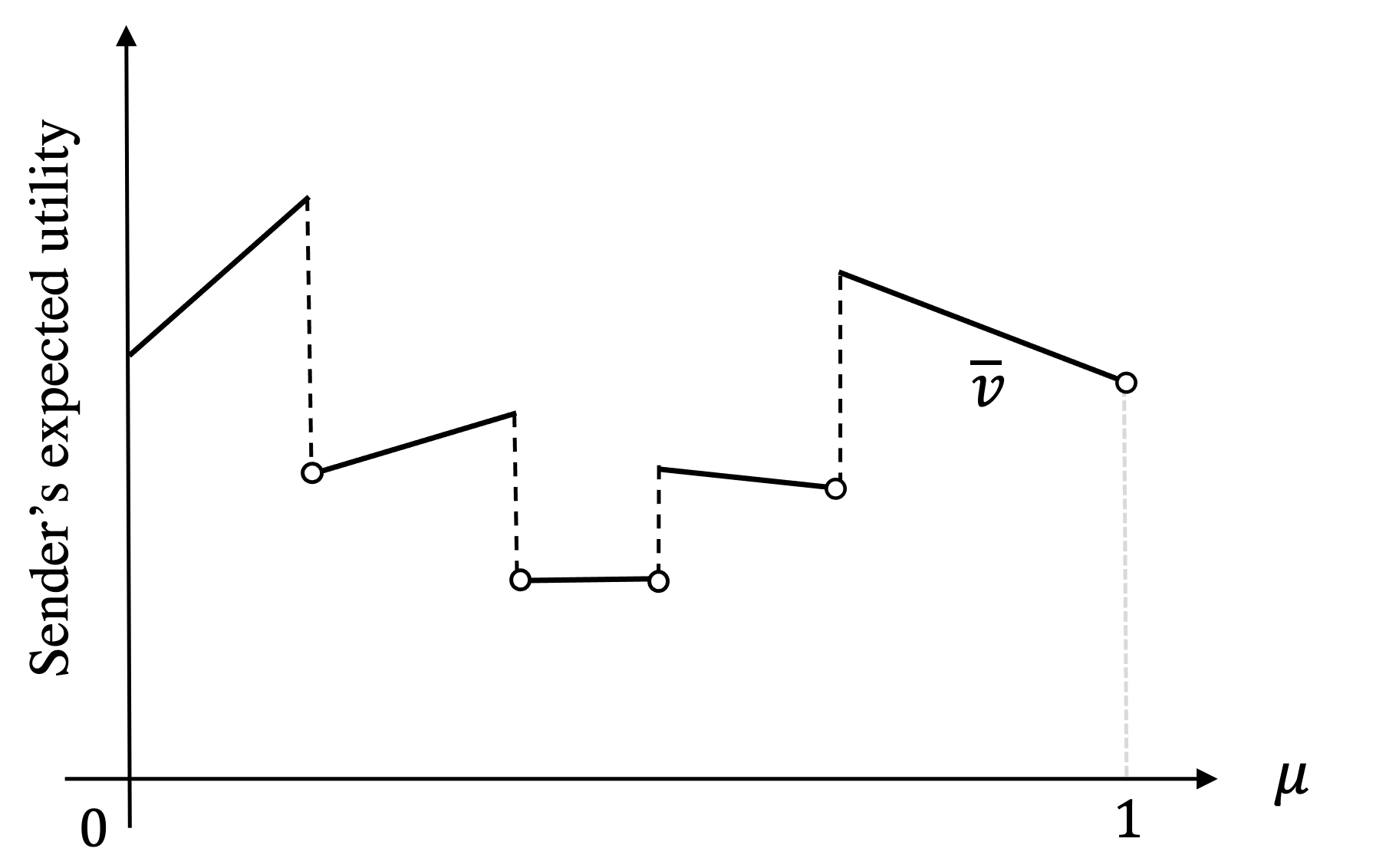}
		\caption{No information transmission.}
		\label{oppoex}
	\end{figure}
	
\subsection{Aligned marginal incentives}	
	Now, we turn to the case where sender and receiver have aligned marginal incentives.
	
	\begin{definition}
		\label{def2}
		\textnormal{An action $a' \in  A$\  \emph{blocks} $a'' \in  A$ if 
			\begin{equation*}
				u_S(a',\mu'') \ge u_S(a'',\mu'')\quad \text{for all } \mu'' \in \{\mu : a'' \in A_R(\mu)\}.
			\end{equation*}
			Action $a' \in  A$ is an \emph{all-blocker} if it blocks all actions in $ A$.
		}
	\end{definition}
	
	According to Definition \ref{def2}, action $a' \in   A$ is 
	an all-blocker if and only if
	\begin{equation*}
		u_S(a',\mu) \ge \overline{v}(\mu)\quad \text{for all } \mu \in [0,1].
	\end{equation*}
	If $a'$ does not block $a''$ and $a''$ does not block $a'$, then the incentive compatibility constraints (\ref{IC-sender}) can be satisfied and there is an IC-PP information structure at some initial belief that will induce these two actions.  

	\begin{definition}\label{def3}
		\textnormal{An action $a' \in   A$ is \emph{worst} if, for all $a'' \in  A$,
			\begin{equation*}
               u_S(a',\mu) \le u_S(a'',\mu)\quad \text{for all } \mu \in [0,1] .
			\end{equation*}
			An action $a' \in   A$ is \emph{best} if, for all $a'' \in  A$,
			\begin{equation*}
                u_S(a',\mu) \ge u_S(a'',\mu)\quad \text{for all } \mu \in [0,1] .
			\end{equation*}
		}
	\end{definition}

	If action $a'$ is worst, the sender prefers any action in $A$ to this action at any belief $\mu$.  It implies that any other action in $A$ blocks $a'$, and $a'$ does not block any other action.  The converse is not true.  Similarly, a best action is necessarily an all-blocker, but an all-blocker need not be best.

	\begin{proposition}
		\label{dominate}
		If the sender and the receiver have aligned marginal incentives, then there is positive information transmission if either of the following holds:
		\begin{itemize}
			\item[(a)] No action is an all-blocker for the sender.
			\item[(b)] No action is worst for the sender. 
		\end{itemize}  
	\end{proposition}
	
	\begin{proof}[Proof of part (a)] 
		For any pair of distinct actions  $a', a'' \in A$, there are four mutually exclusive possibilities: (1) $a'$ blocks $a''$ and $a''$ does not block $a'$; (2) $a''$ blocks $a'$ and $a'$ does not block $a''$; (3) neither action blocks the other; or (4) each action blocks the other.  Case (4) is impossible under aligned marginal incentives. 
		We claim that at least one pair of actions in $ A$ must fall under case (3).  Suppose this claim is false, so that case (1) and case (2) mutually exhaust all possibilities on  $A$.  Then the binary relation ``block'' on  $A$ would be reflexive, complete, and antisymmetric.  In the next paragraph, we show that it would also be transitive, and therefore ``block'' would be a total order on the finite set $ A$, which would further imply that there is a maximal action  on  $ A$, i.e., an all-blocker action exists in $ A$.  This is a contradiction, and therefore we conclude that at least one pair of actions, $a'$ and $a''$, must fall under case (3).  This pair of actions is strictly IC-PP because the complement of Definition \ref{def2} imposes strict inequality.
		Thus, an information structure that induces these two actions will improve the sender's payoff when, for example, the prior belief is in the interior of $\{\mu : a' \in A_R(\mu)\}$. 
		
		To see 
		why transitivity holds under the premise that cases (1) and (2) mutually exhaust all possibilities on $ A$,
		consider $|A|\ge 3$. (If $|  A|=2$, it is immediate that ``block'' is a total order as the two actions are comparable.)
		Suppose $a$ blocks $b$ and $b$ blocks $c$, and let $\mu_a$, $\mu_b$ and $\mu_c$ be three distinct beliefs at which these three actions are respective best responses.  (a) Suppose $\mu_a < \mu_b$.  (a)(i) If $\mu_c < \mu_b$, then $a$ blocks $b$ implies $u_S(a,\mu_b) \ge u_S(b,\mu_b)$. 
		Aligned marginal incentives (supermodularity  of $u_S(\cdot,\cdot)$) then imply
		$u_S(a,\mu_c) \ge u_S(b,\mu_c) \ge u_S(c,\mu_c)$, where the last inequality follows because $b$ blocks $c$.  Since this argument holds for any $\mu_c < \mu_b$, we conclude that $a$ blocks $c$.  (a)(ii) If $\mu_c > \mu_b$, then $b$ blocks $c$ implies $u_S(b,\mu_c) \ge u_S(c,\mu_c)$.  
		Aligned marginal incentives then 
		imply
		that there exists $\mu_a \in \{\mu: a \in A_R(\mu)\}$ such that $u_S(a,\mu_a) > u_S(b,\mu_a) \ge u_S(c,\mu_a)$, where the first inequality follows because $b$ does not 
		block $a$.  This shows that $c$ does not block $a$.  Since cases (1) and (2) are mutually exhaustive possibilities under the supposition that no pair of actions falls under case (3), $a$ blocks $c$ whenever $c$ does not block $a$.  The analysis of (b), where $\mu_a > \mu_b$, is symmetric.  In both cases, 
        if $a$ blocks $b$ and $b$ blocks $c$, then $a$ blocks $c$.
	\end{proof}

	The proof of part (b) of Proposition \ref{dominate} 
	involves finding IC-PM pairs and is more tedious; we leave it to  the Appendix.  Figure \ref{aligned} provides two examples to illustrate this proposition.
	The left panel of Figure \ref{aligned} shows a case where there is no all-blocker action.
	By proposition \ref{dominate}(a), there must exist a pair of distinct actions such that neither action blocks the other action. In the figure, the random posterior $\{0,\mu_1\}$ is IC-PP for $a_{-1}$ and $a_0$ and it improves the sender's payoff at the prior belief $\mu_0$.
	
	\begin{figure}[t]
		\centering
		\includegraphics[width=15cm]{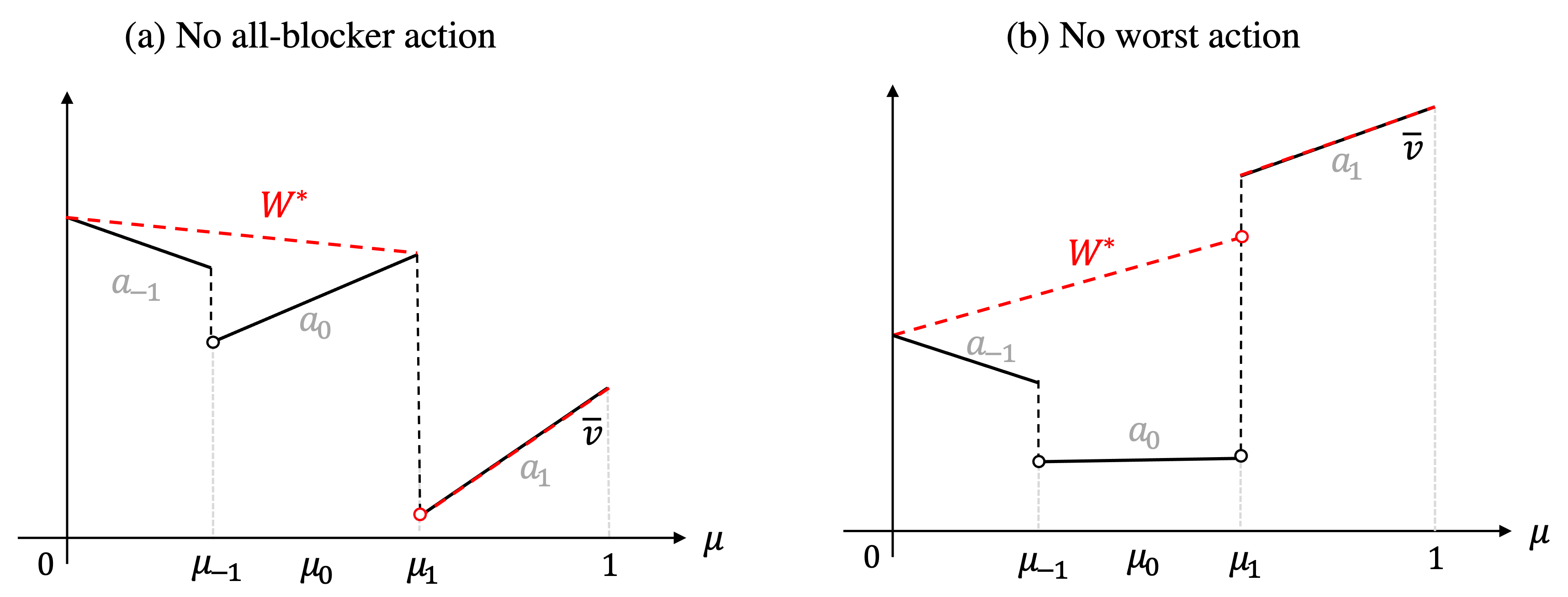}
		\caption{Positive information transmission.}
		\label{aligned}
	\end{figure}
	
	Next, consider the right panel, where no worst action exists.  
	In this example, the sender prefers $a_1$ to $a_{-1}$ to $a_0$ at belief 0. Therefore, we can find a randomization 
	$\alpha_1^{PM} \in \Delta A_R(\mu_1)$ (shown by the red dot)
	such that the sender is indifferent between $a_{-1}$ and $\alpha_1^{PM}$ at belief 0. Moreover, from aligned marginal incentives, 
	we have $ m_S(a_{-1})<m_S(\alpha_1^{PM})$, 
	implying that the sender must strictly prefer $\alpha_1^{PM}$ to $a_{-1}$ at belief $\mu_1$.
	Hence, the random posterior $\{0,\mu_1\}$ is IC-PM and induces $a_{-1}$ and $\alpha_1^{PM}$ at the two beliefs.  This random posterior improves the sender's payoff when the prior belief is $\mu_0$.

Conditions (a) and (b) in Proposition \ref{dominate} are each sufficient for positive information transmission, but neither of them is necessary. For example, action $a_1$ in Figure \ref{aligned}(a) is a worst action, and action $a_1$ in Figure \ref{aligned}(b)  is an all-blocker action.  That is, there can still be informative information transmission when an all-blocker action or a worst action exists for the sender.

	Proposition \ref{oppoprop} and \ref{dominate} suggest
	that the alignment of marginal incentives between sender and receiver is important for generating incentive compatibility. Given aligned marginal incentives, the alignment of preference ranking over actions, on the other hand,  is less important.  
We now elaborate on it. 
 
	\begin{definition}\label{def4}
		\textnormal{
			Sender's preferences are \textit{ordinally state-independent} if, for every $a',a'' \in   A$,
			\begin{equation*}
				u_S(a',1) > u_S(a'',1) \iff u_S(a',0) > u_S(a'',0).
			\end{equation*}
		}
	\end{definition} 
	
	This definition implies that the sender's ranking over actions is the same at any $\mu \in [0,1]$.  It is a generalization of transparent motives (i.e., state-independent preferences) because this class of preferences does not require $m_S(a)$ to be equal to 0 for all $a$.
	
	Given the labeling we adopt on the action space, the receiver's ranking over action in state 0 is decreasing in the index of actions,
	and is increasing in the index of actions when the state is $1$.
	A sender with ordinally state-independent preferences can have arbitrary ranking over actions, 
	though with aligned marginal incentives, his marginal incentive 
is higher for a higher action. 
	
	\begin{proposition}
		\label{state-inde}
		Suppose sender and receiver have aligned marginal incentives, and the sender's preferences are ordinally state-independent. Information transmission is positive if and only if the sender's ranking of actions is non-monotone in the index of actions. 
	\end{proposition}

	Ordinal state-independence implies that there does not exist an IC-PP 
random posterior, because for any two distinct pure actions $a'\ne a''$, either $u_S(a',\mu) > u_S(a'',\mu)$ for all $\mu \in [0,1]$, or the opposite (strict) inequality holds for all $\mu \in [0,1]$.    Therefore, incentive compatibility  necessarily requires the receiver's randomization between a higher-ranking action and a lower-ranking action. Provided marginal incentives are aligned, we show the existence of an 
IC-PM or IC-MP 
random posterior except in the special case where the sender's ranking over actions is identical to the receiver's ranking in one of the states. We provide the complete proof in the Appendix.

\subsection{Neither opposite nor aligned marginal incentives}
The definitions of opposite and aligned marginal incentives require the same ordering of marginal incentives for all actions, i.e., with opposite (aligned) marginal incentives, the slope of the sender's piecewise indirect value function is decreasing (increasing). Nevertheless, having a monotone ordering of marginal incentives for all actions is restrictive and is not necessary for generating credibility. 
The role of the receiver's randomization is to smooth the sender's marginal incentives and thereby generate the correct ordering we need to impose for incentive compatibility. 
We illustrate this with an example. 

\textbf{\emph{An example.}}  
A policy maker is seeking public support to change the status quo from no carbon tax ($a_0$) to a policy of either a small carbon tax ($a_1$) or high carbon tax ($a_2$).
There are two states: low risk of climate change (state 0) and high risk of climate change (state 1). 
In state 0, the public's preference is $a_0\succ a_1\succ a_2$. 
In state 1, the public's preferences are reversed: $a_2\succ a_1\succ a_0$. At the prior belief $\mu_0$, the public prefers to remain at the status quo. The policy maker's preference is $a_1\succ a_0\succ a_2$ in state 0, and is $a_1\succ a_2\succ a_0$ in state 1, indicating 
that $a_1$ is a best action, and that the policy maker and the public have aligned interests for extreme actions $a_0$ and $a_2$, which implies  $m_S(a_0)<m_S(a_2)$. 
The low carbon tax is a best action because the policy maker finds it most politically expedient, but he also considers that small taxation has a limited impact on greenhouse gas emissions and may foster a false sense of accomplishment, potentially postponing the public's pursuit of alternative measures, which would be especially harmful if climate risk is high.  
We assume that $m_S(a_1)<m_S(a_0)$.
Thus, the policy maker and the public exhibit opposite marginal incentives for $a_0$ and $a_1$ but aligned marginal incentives for $a_0$ and $a_2$. 
The policy maker can hire a research group to conduct specific research to learn about the state. However, the outcomes of this research are unverifiable and, therefore, considered as cheap talk messages to the public.

	\begin{figure}[t]
		\centering
		\includegraphics[width=7cm]{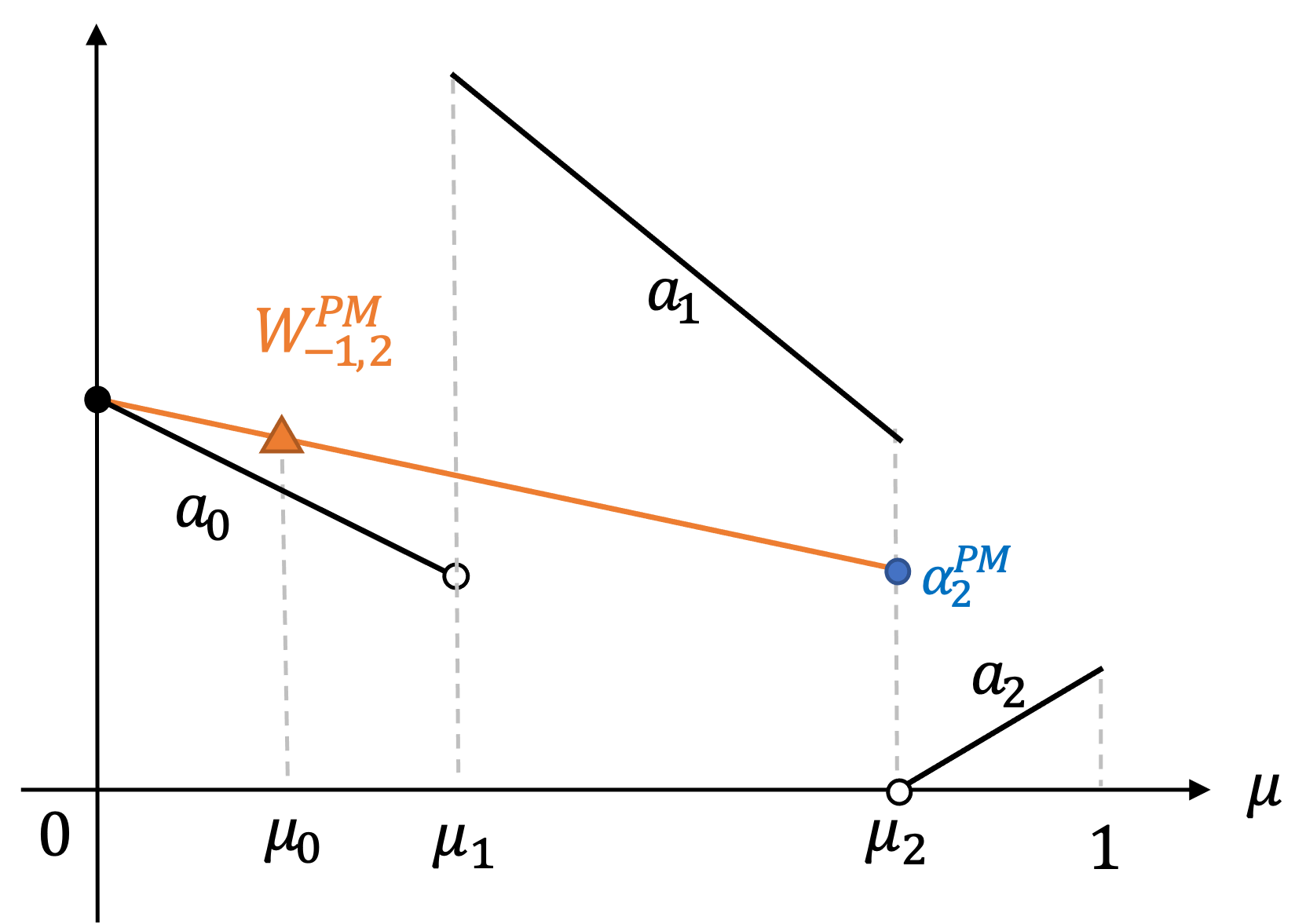}
		\caption{Neither opposite nor aligned marginal incentives. \small Opposite marginal incentives between $a_0$ and $a_1$; aligned marginal incentives between $a_0$ and $a_2$. }
		\label{climate}
	\end{figure}

Figure \ref{climate} describes the policy maker's indirect value function in this scenario. He achieves the highest equilibrium payoff by designing the optimal random posterior $\{0,\mu_2\}$, and the receiver randomizes between the two taxation magnitudes $a_1$ and $a_2$ upon receiving message $\mu_2$. This randomization between the sender's most-preferred action $a_1$ and the least-preferred action $a_2$ at belief $0$ helps mitigate the sender's incentive to misreport a higher belief $\mu_2$ when his true private belief is $0$. In addition, this randomization smooths the sender's marginal incentive, and thereby restores a proper ordering of
marginal incentives (i.e., $m_S(a_0)<W'_{-1,2}(\cdot;a_0,\alpha^{PM}_2)=m_S(\alpha^{PM}_2)$) to satisfy the incentive compatibility at belief $\mu_2$.
Overall, with state-dependent preferences, the receiver's randomization has two roles: one is to average the utilities, and the other is to average the marginal incentives.

\section{One-Sided Common Interest}
\label{s:commoninterests}

	In many situations, the sender and the receiver may have common interests in one state but conflicting interests in another state. By this, we mean that the receiver's optimal action in one state is also the sender's most-preferred action in that state (their rankings over other actions in that state can be different).

	\begin{definition}\label{def5}
		\textnormal{
			Sender and receiver have \emph{common interest in one state} if, for $\theta=0$ or $\theta=1$,
			\begin{equation*}
				u_S(a,\theta) \ge u_S(a',\theta) \quad \text{for all $a\in A_R(\theta)$ and all $a' \in  A$}.
			\end{equation*}
		}
	\end{definition} 
	
	With common-interest in one state, we can disentangle the sender's trade-off between acquiring more information and alleviating the conflicts of interest. On the information side, the sender may want to reveal more information about the common interest state---instead of pooling the common-interest state with the other state---so that he can make the correct recommendation more often. On the side of conflicts of interest, since sender and receiver prefer the same action under the common-interest state, revealing it can further increase the sender's ex-post payoff in that  common-interest state and thereby 
on average increase the sender's ex-ante payoff. 
The proposition below confirms this intuition. It shows that under some mild conditions, the optimal information structure generates a conclusive signal on the common-interest state.

	\begin{proposition}
		\label{commoninterest}
		Let the common-interest state be state 0, and let the optimal action corresponding to that state be $a_{-J}$. 
  If there exists an action $a_k\in \{A_R(\mu): \mu\in(\mu_0,1]\}$ such that $a_{-J}$ does not block $a_k$, then  $0\in \operatorname{supp} P^*$.
	\end{proposition}
	
	\begin{proof}
        Since $a_{-J}$ does not block $a_k$, we have
		$u_S(a_k,\mu_{k+1})>u_S(a_{-J},\mu_{k+1})$. Let $\overline{a}_{k+1}$ be the sender-preferred action in $A_R(\mu_{k+1})$. Then $u_S(\overline{a}_{k+1},\mu_{k+1})\ge u_S(a_k,\mu_{k+1})>u_S(a_{-J},\mu_{k+1})$. From the definition of common interest in state 0, $u_S(a_{-J},0)\ge u_S(\overline{a}_{k+1},0)$. Therefore,  the random posterior with support $\{0,\mu_{k+1}\}$ is IC-PP if the receiver optimally chooses between $a_{-J}$ and $\overline{a}_{k+1}$.  
		
		By Proposition \ref{withoutloss}, it is without loss of generality to only consider information structures that generate posteriors that are in the set of boundary beliefs $B$.  Consider an incentive compatible random posterior $P'$ with support $\{ \mu_{-j},\mu_{k'}\}$ that induces $\alpha_{-j}\in A_R(\mu_{-j})$ and $\alpha_{k'}\in A_R(\mu_{k'})$.  Consider another random posterior $P$ with support $\{0,\mu_{k+1}\}$ 
		that induces actions $a_{-J}$ and $\overline{a}_{k+1}$ at these beliefs.  There are two possibilities.

		Case (1) $\mu_{k'}=\mu_{k+1}$. 
		Since $P'$ is incentive compatible, the payoff from this information structure is
		\begin{align*}
			W_{-j,k'}(\mu_0;\alpha_{-j},\alpha_{k'}) &= \mathbb{E}_{P'}\left[\max\left\{ \mathbb{E}_{\alpha_{-j}}[u_S(a,\mu)], \mathbb{E}_{\alpha_{k'}}[u_S(a,\mu)] \right\}\right] \\
			& \le \mathbb{E}_{P}\left[\max\left\{ \mathbb{E}_{\alpha_{-j}}[u_S(a,\mu)], \mathbb{E}_{\alpha_{k'}}[u_S(a,\mu)] \right\}\right] \\
			& \le  \mathbb{E}_{P}\left[\max\left\{  u_S(a_{-J},\mu), \mathbb{E}_{\alpha_{k'}}[u_S(a,\mu)] \right\}\right] \\
			& \le  \mathbb{E}_{P}\left[\max\left\{  u_S(a_{-J},\mu),  u_S(\overline{a}_{k+1},\mu) \right\}\right] \\
			& = W_{-J,k+1}(\mu_0;a_{-J},\overline{a}_{k+1}).
		\end{align*}
		The first inequality follows from the fact that $P$ is a mean-preserving spread of $P'$; therefore there is positive information value when the receiver's action space is fixed:  belief $\mu_{k+1}$ is realized more often under $P$ and the sender can correctly recommend $\alpha_{k'}$ instead of $\alpha_{-j}$ at $\mu_{k+1}$.  The second inequality follows from 
common-interest at state 0, $\mathbb{E}_{\alpha_{-j}}[u_S(a,0)]\le u_S(a_{-J},0)$---revealing state 0 increases the sender's payoff as the receiver will take a more favorable action when state 0 is realized.  The third inequality comes from $\mathbb{E}_{\alpha_{k'}}[u_S(a,\mu_{k+1})]\le u_S(\overline{a}_{k+1},\mu_{k+1})$. The last equality comes from 
		the fact
		that the random posterior with support $\{0,\mu_{k+1}\}$ is IC-PP for $a_{-J}$ and $\overline{a}_{k+1}$. 
		
		Case (2) $\mu_{k'}\ne \mu_{k+1}$. 
		If the information structure $\{0,\mu_{k'}\}$ that induces $a_{-J}$ and $\alpha_{k'}$ is incentive compatible, then the same argument provided in case (1) shows that this information structure will give a higher payoff to the sender than 
		does $P'$.  So we only need to consider the case that $\{0,\mu_{k'}\}$ is not incentive compatible.  In this case, because $a_{-J}$ is the sender's most-preferred action in state 0, incentive compatibility can fail only when $u_S(a_{-J},\mu_{k'}) >\mathbb{E}_{\alpha_{k'}} [u_S(a,\mu_{k'})]$ (i.e., the sender prefers $a_{-J}$ to $\alpha_{k'}$ at belief $\mu_{k'}$). 
		Moreover,  
		since the sender prefers $\alpha_{k'}$ to $\alpha_{-j'}$ at belief
		$\mu_{k'}$ (incentive compatibility), by transitivity he prefers $a_{-J}$ to $\alpha_{-j}$ at belief $\mu_{k'}$.  He also prefers $a_{-J}$ to $\alpha_{-j}$ at belief $0$.  Because preferences are linear in beliefs, this implies that he prefers $a_{-J}$ to $\alpha_{-j}$ at belief $\mu_{-j}$.  
		Therefore,
		\begin{align*}
			\mathbb{E}_{P'}\left[\max\left\{ \mathbb{E}_{\alpha_{-j}}[u_S(a,\mu)], \mathbb{E}_{\alpha_{k'}}[u_S(a,\mu)] \right\}\right] & <  
			\mathbb{E}_{P'}\left[ u_S(a_{-J},\mu)\right]	\\
			&=u_S(a_{-J},\mu_0)\\
			&\le \mathbb{E}_{P}\left[\max\left\{ u_S(a_{-J},\mu), u_S(\overline{a}_{k+1},\mu) \right\}\right].
		\end{align*}
		The first inequality follows from the fact that $a_{-J}$ is strictly better than $\alpha_{-j}$  and $\alpha_{k'} $ at 
		belief $\mu_{-j}$ and belief $\mu_{k'}$, respectively.
		The last inequality follows from the fact that the information structure $P$ is incentive compatible for  $a_{-J}$ and $\overline{a}_{k+1}$.  
		Therefore there is a positive information value as the sender can correctly recommend $\overline{a}_{k+1}$ instead of $a_{-J}$ at belief $\mu_{k+1}$. 
	\end{proof}
	
	Proposition \ref{commoninterest}
	implies that as long as $u_S(a_{-J},\mu) \le \overline{v}(\mu)$ for some $\mu > \mu_0$, 
	the	support of the optimal random posterior contains $0$.
	If it also contains $1$, then the optimal experiment reveals perfect information.  If it does not 
	contain $1$, the optimal experiment will generate a conclusive signal of the common-interest state.  In other words, the underlying Blackwell experiment corresponding to this optimal random posterior will produce a signal that reveals the common-interest state 0 with probability strictly less than 1 when the true state is 0, and never produces a signal that would suggest the state is 0 when the true state is 1.  This means that the ex-ante probability that the receiver takes action $a_{-J}$ under the optimal information structure cannot exceed	$1-\mu_0$.

	\section{Informativeness Compared to Bayesian Persuasion}
	\label{newsection}
	
	In general, the optimal experiment when the sender has no commitment power can be more or less informative than (or not Blackwell-comparable to) the optimal experiment chosen when the sender can commit to truthfully revealing the outcome of the experiment.  For example, when sender and receiver have opposite marginal incentives, Proposition \ref{oppoprop} shows that the optimal experiment in our setup is a totally uninformative experiment, while the optimal experiment with full commitment is typically non-degenerate, as the concave envelope of $\overline{v}(\cdot)$ does not coincide $\overline{v}(\cdot)$ itself.
	
	For an example in which the optimal experiment in our setup is more informative than that in a model with full commitment, consider the case where there is a best action $a_n$ that the sender prefers the most in both states.  Let $a_n$ be the receiver's best response when the belief is in the interval $[\underline{I}_n,\overline{I}_n]$. 
	Let the prior belief $\mu_0$ be lower than $\underline{I}_n$. 
	The lesson we learn from \cite{kamenica2011bayesian} is that if the optimal experiment with full commitment induces $a_n$ and some other action, it maximizes the ex-ante probability that $a_n$ will be taken by inducing the smallest posterior belief $\underline{I}_n$ that is just enough to induce the receiver to choose $a_n$. 
	When the sender lacks commitment power in communication, inducing the pure action $a_n$ 
	is not incentive compatible.
	However, it may be incentive compatible to induce a mixture action between $a_n$ and $a_{n+1}$ at belief $\overline{I}_n$.  Because $\overline{I}_n$ is farther from $\mu_0$ than $\underline{I}_n$ is from $\mu_0$, the resulting experiment is 
	more informative than the optimal experiment under full commitment.  
	The next proposition specifies the precise conditions for an analogous argument to be valid.

	\begin{proposition}\label{compwithkg}
		Assume that sender and receiver have aligned marginal incentives and $|A|\ge 3$. If an action $a_n$ ($n \ne -J,K$) is (strictly) best, then the optimal information structure in our model is (strictly) more informative than the optimal experiment under full commitment for some prior belief. 
	\end{proposition}	
	
	The proof of this proposition is in the Appendix. Consider Figure \ref{kg}, $a_n$ is the best action for the sender	and the dotted red envelope is the concave envelope of the sender's value function. The optimal experiment under full commitment at the prior belief $\mu_0$ has support $\{\underline{I}_{n'},\underline{I}_n\}$. 
	The optimal experiment under full commitment at belief 
	$\mu'_0$ has support $\{\overline{I}_{n},\overline{I}_{n^{''}}\}$. 
	
	\begin{figure}[t]
		\centering
		\includegraphics[width=15.5cm]{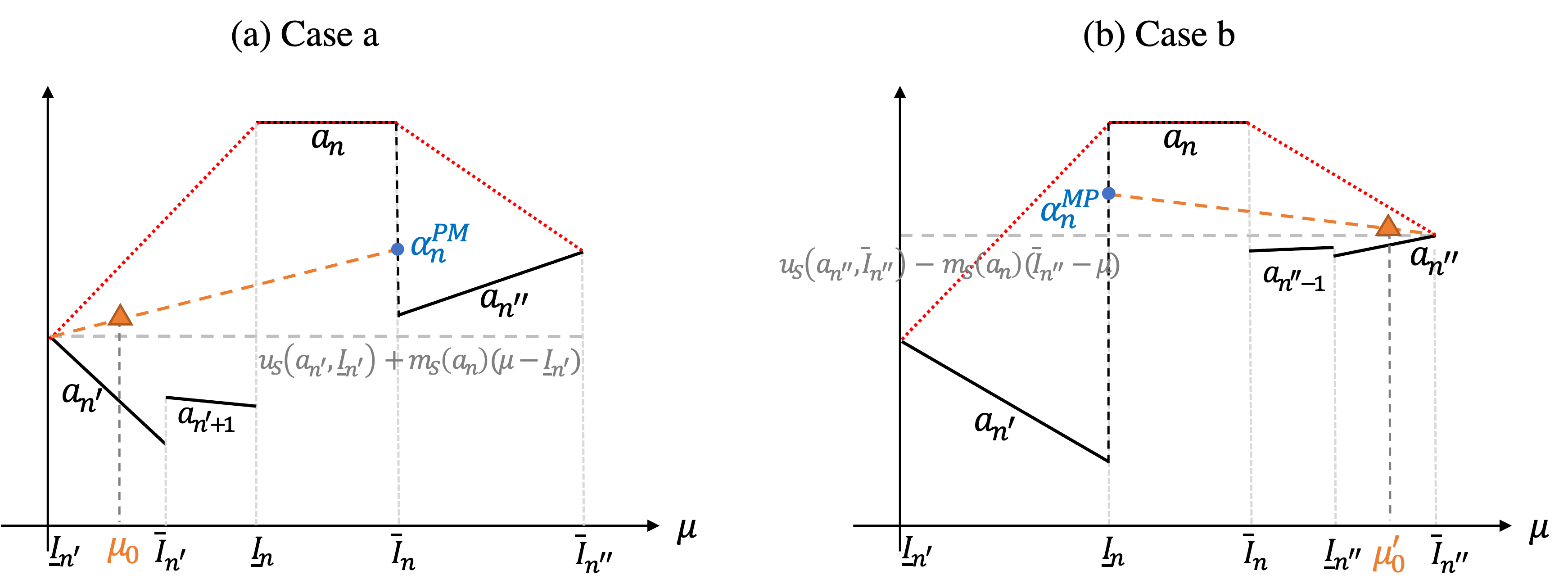}
		\caption{Optimal experiment 
			with and without commitment.}
		\label{kg}
	\end{figure}

	In the left panel (case a),
	the sender with belief $\underline{I}_{n'}$ strictly prefers $a_n$ over $a_{n'}$ over $a_{n''}$. It implies that there exists a randomization $\alpha^{PM}_n$ between $a_n$ and $a_{n''}$ such that the experiment with support $\{\underline{I}_{n'},\overline{I}_n\}$ is IC-PM. Recall that with aligned marginal incentives, $m_S(\alpha^{PM}_n)>m_S(a_{n})$.
	Therefore the expected payoff from such IC-PM experiment (the orange triangle) is strictly higher than $u_S(a_{n'},\underline{I}_{n'})+m_S(a_n) (\mu_0-\underline{I}_{n'})$ (lying on the gray dashed line). Moreover, with aligned marginal incentives, any incentive compatible experiment that induces $a_{n'}$ and some action smaller than $a_n$ can only lead to an expected payoff strictly below $u_S(a_{n'},\underline{I}_{n'})+m_S(a_n) (\mu_0-\underline{I}_{n'})$. 
	For example, in Figure \ref{kg}(a), the experiment with support $\{\underline{I}_{n'},\underline{I}_{n}\}$ is IC-PP for $a_{n'}$ and $a_{n'+1}$. However, the sender's expected payoff from it is bounded by the gray dashed line because the slope of the sender's expected payoff is smaller than the marginal incentives of $a_{n'+1}$ (implied by Lemma \ref{l:IC}) which is smaller than $m_S(a_n)$. 
	Thus, in this case, under the prior belief $\mu_0$, the 
	optimal experiment in our model is more informative than that under full commitment.

	It is possible that the sender with belief $\underline{I}_{n'}$ strictly prefers all actions higher than $a_n$ over $a_{n'}$, so that we cannot find an incentive compatible experiment that is more informative than $\{\underline{I}_{n'},\underline{I}_n\}$. This happens in the right panel (case b). However, given the assumption of aligned marginal incentives, there must exist 
	in this case two actions (weakly) smaller than $a_n$ such that the sender with belief $\overline{I}_{n''}$ prefers one over $a_{n''}$ over the other one. In Figure \ref{kg}(b), type-$\overline{I}_{n''}$ sender prefers $a_n$ over $a_{n''}$ over $a_{n'}$. With similar reasoning as in case (a), under the prior belief $\mu'_0$, the optimal experiment in our model is more informative than that under full commitment.

\section{Canonical Cheap Talk}
	\label{s:cheaptalk}
	
	In this section, we discuss the connection between our model and the canonical cheap talk model under binary state space and finite action space. We introduce a modification of Algorithm 1 to find the highest equilibrium payoff for the sender in a canonical cheap talk game. 
	
	In the canonical cheap talk game, the sender is initially perfectly informed about the true state. His reporting strategy is essential for generating credible information transmission. 	It is obvious that for every equilibrium in the canonical cheap talk game, there is a corresponding game in our model inducing the same equilibrium outcome---namely,
	the sender commits to the information structure that is his reporting strategy in the canonical cheap talk game, and then truthfully reports his private information outcomes. 
	
	Conversely, pick a truth-telling equilibrium in our game. To ensure that its outcome is also feasible in the canonical cheap talk, we need to verify an additional constraint: 	while fixing the receiver's decision rule to be the same as in the truth-telling equilibrium, the sender cannot gain from deviating to a more informative information structure. 
    If this constraint fails, using the information structure in our game as the reporting strategy is not incentive compatible for a sender who knows the true state.

 \begin{figure}[t]
		\centering
		\includegraphics[width=15cm]{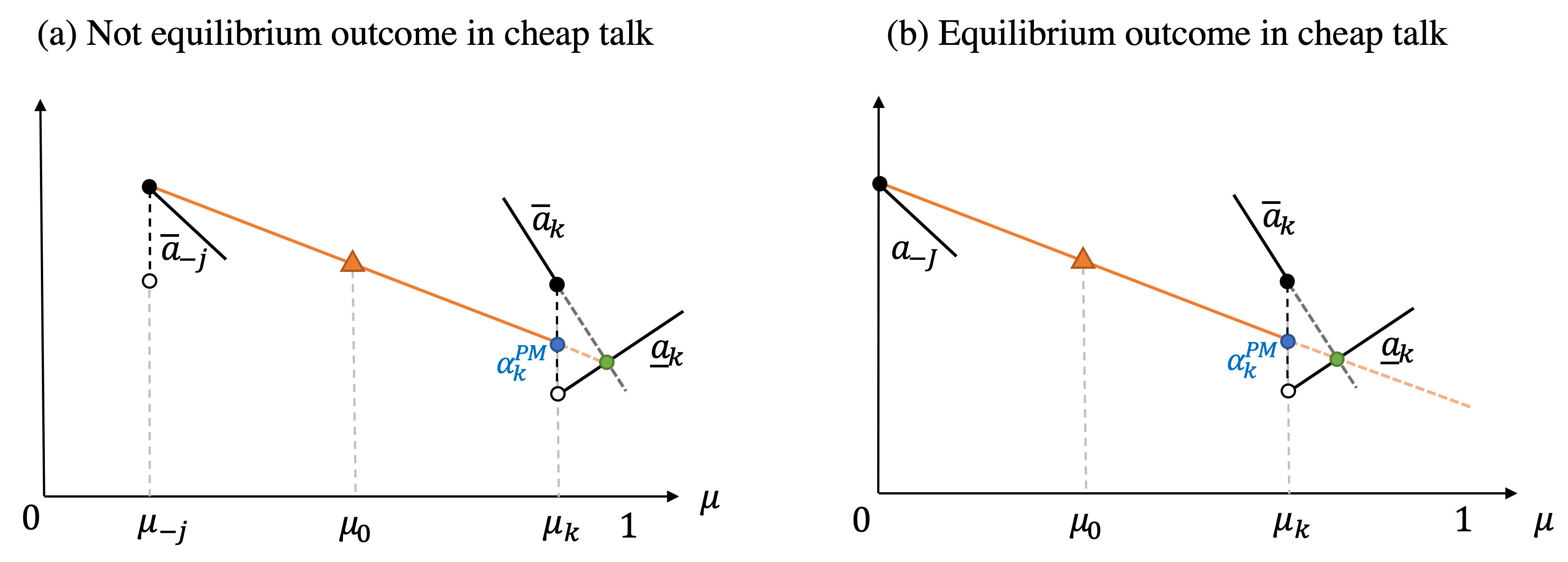}
		\caption{Comparing with canonical cheap talk}
		\label{cheaptalk}
	\end{figure}	
 
	To illustrate the logic behind this constraint,
	Figure \ref{cheaptalk}(a) provides an example of one-sided randomization (PM).
	The random posterior $\{\mu_{-j},\mu_k\}$ is IC-PM  
	that the sender
	is indifferent between mixed action $\alpha^{PM}_k$ and pure action $\overline{a}_{-j}$ at belief $\mu_{-j}$, and strictly prefers $\alpha^{PM}_k$ to $\overline{a}_{-j}$ at belief $\mu_k$.  Because both $\mu_{-j}$ and $\mu_k$ are interior and preferences are linear in beliefs, this in turn implies that the sender strictly prefers  $\alpha^{PM}_k$ to $\overline{a}_{-j}$  at belief 1 and strictly prefers $\overline{a}_{-j}$ to  $\alpha^{PM}_k$ at belief 0.
	The equilibrium outcome induced by this one-sided randomization cannot be sustained as an equilibrium outcome in the canonical cheap talk game. 
	To produce an outcome that induces interior beliefs $\mu_{-j}$ and $\mu_k$, the sender must adopt a reporting strategy that recommends both actions with positive probabilities in each state, but this is not incentive compatible for a sender who knows the true state.  
	To put it slightly differently, the information structure with support $\{\mu_{-j},\mu_k\}$ cannot be sustained as an equilibrium outcome if the sender cannot commit to this experiment, because he can gain from deviating to learn more about the state.

	Figure \ref{cheaptalk}(b) modifies the example to show when one-sided randomization can be supported as an equilibrium outcome in the canonical cheap talk game. 
 The only difference is that, in Figure \ref{cheaptalk}(b), the pure action $\overline{a}_{-j}=a_{J}$ is taken at degenerate belief $0$.  To produce the random posterior $\{0,\mu_k\}$ under cheap talk, the sender must adopt a reporting strategy which recommends the mixed action $\alpha^{PM}_k$ only in state 1, and recommends both actions with positive probability in state 0.  
	Since the sender is indifferent between these two actions in state 0, such a reporting strategy is indeed incentive compatible and will produce an interior belief $\mu_k$ for the receiver upon getting the recommendation to choose $\alpha^{PM}_k$.

	This example demonstrates that an IC-PM or IC-MP outcome in our model cannot be supported as an equilibrium outcome of the cheap talk game if the pure action is chosen at an interior belief 
and the sender has a strict preference at the other belief,\footnote{
If the sender is indifferent at the other belief that induces the mixed action, then such one-sided randomization can be an equilibrium outcome in a cheap talk game. However, it is covered by IC-MM. }  
but it is an equilibrium outcome of the cheap talk game if the pure action is chosen at degenerate beliefs (0 or 1). In other words, when looking for 
pure action in an IC-PM or IC-MP equilibrium of the canonical cheap talk game, we only need to consider action $a_{-J}$ or $a_K$ respectively.  Similarly, when looking for an IC-PP equilibrium in the canonical cheap talk game, we only need to consider the actions $a_{-J}$ and $a_{K}$.
On the other hand, every IC-MM outcome in our model can be supported as an equilibrium outcome in the canonical cheap talk game. Because by construction, the sender is indifferent between a pair of IC-MM actions at beliefs $\mu_{-j},\mu_k$, and therefore he is also indifferent between them at beliefs 0 and 1.
	
	This suggests that the sender-optimal equilibrium in the cheap talk model can be obtained from a straightforward modification of Algorithm 1. 

\textbf{\emph{Algorithm 2:}}
	\setlist{nolistsep}
	\begin{enumerate}
		
		\item For the pair $(-J-1,K+1)$, compute $W_{-J-1,K+1}(\mu_0;a_{-J},a_{K})$. If this pair is IC-PP, let $W^1=W_{-J-1,K+1}(\mu_0;a_{-J},a_{K})$;  otherwise let $W^1=\bar{v}(\mu_0)$. 
		\vspace{0.2cm}
		
		\item For every pair $(-J-1,k)$ where $k\in\{1,...,K+1\}$, compute $W_{-J-1,k}(\mu_0;a_{-J},\alpha_k^{PM})$ and rank these values from highest to lowest. Starting with the pair with the highest value, verify whether it is IC-PM or not.  Stop the first time when an IC-PM pair is found. Assign $W^{(a)}=W_{-J-1,k}^{PM}$ for such pair. If none of them is IC-PM, assign $W^{(a)}=\overline{v}(\mu_0)$. 
Symmetrically, for every pair $(-j,K+1)$ where $j\in\{-J-1,...,-1\}$, go through a similar procedure in the case for MP. Assign $W^{(b)}=W_{-j,K+1}^{MP}$ the first time an IC-MP is found.  If none of them is IC-PM, assign $W^{(b)}=\overline{v}(\mu_0)$.
Set $W^2=\max\{W^{(a)},W^{(b)}\}$. 
		\vspace{0.2cm}

		\item For every pair $(-j,k)\in 
		\{-J,\ldots,-1\}\times \{1,\ldots,K\}$, compute $W_{-j,k}(\mu_0;\alpha_{-j}^{MM},\alpha_{k}^{MM})$ and rank these values from the highest to lowest.
Starting with the pair with the highest value, verify whether it is IC-MM or not. Stop the first time an IC-MM pair is found and assign $W^3=W_{-j,k}^{MM}$ for such pair. If none of them is IC-MM, assign $W^3=\overline{v}(\mu_0)$. 
		\vspace{0.2cm}

		\item Assign $W^c(\mu_0)=\max\{W^1,W^2,W^3\}$. The random posterior with support $\{\mu_{-j},\mu_k\}$ corresponds to the pair that yields $W^c(\mu_0)$ is the sender's reporting strategy in the optimal equilibrium. 
  \vspace{0.2cm}
	\end{enumerate}

 \begin{corollary}
     Algorithm 2 determines the highest equilibrium payoff for the sender in a canonical cheap talk game (with binary states and finite actions). 
 \end{corollary}

It is immediate that sender's commitment on the information structure is valuable if and only if $W^*(\mu_0)>W^c(\mu_0)$, i.e., when the two algorithms produce different outcomes.

    \section{Discussion} 
	\label{s:discussion}
	
	The model in this paper has a close relation with a particular scheme of mediated communication, in which a mediator maximizes the ex-ante welfare of an informed sender \citep{Salamanca}. Specifically, a perfectly informed sender sends a message about his private information to a mediator. The mediator then communicates 
	a message to the receiver according to a noisy reporting rule that the mediator commits to at the beginning of the game. After receiving the message from the mediator, the receiver takes an action. If we consider the mediator's reporting rule as a mapping from the sender's private information to a distribution of action recommendations,
	this rule can be interpreted as an information structure.
	The incentive constraints for this mediated communication 
	game are imposed at the \emph{ex ante} stage, which require
	every type of sender who perfectly knows the state
	to report his private information truthfully 
	before observing the message sent by the mediator. 
	
	In contrast, the sender in our model is uninformed when he
	commits to an information structure, and then reports his private information to the receiver after observing the outcome of the experiment.  Therefore, our model requires \emph{interim-stage} incentive constraints, such that the sender with an interim belief derived from the observed outcome prefers to report his private information truthfully. In spite of this difference, if our sender and the mediator in \cite{Salamanca} commit to the same information structure in equilibrium, then both models will yield the same equilibrium outcomes. 
	
	Interestingly, under binary state space, 
	the highest equilibrium payoff 
	$W^*(\mu_0)$ that the sender can achieve in our model
	is always weakly lower than
	the maximum ex-ante welfare of the sender (i.e., evaluated at $\mu_0$ before the sender becomes perfectly informed) in \cite{Salamanca} for any $\mu_0$.	
	
	To see this, suppose the optimal random posterior in our sender-receiver game has support $\{\mu',\mu''\}$ and the receiver optimally chooses  $a'\in A_R(\mu')$ and $a''\in A_R(\mu'')$ at the respective beliefs (the same argument will go through if the receiver takes mixed strategy). 
	Without loss of generality, let $\mu' < \mu''$.  Then
	incentive compatibility constraints (\ref{IC-sender}) in our model
	implies that the following also holds:
	\begin{equation*}
		u_S(a',0) \ge u_S(a'',0),\qquad
		u_S(a'',1) \ge u_S( a',1).
	\end{equation*}	
	This means that it is incentive compatible for an informed sender to truthfully report his private information (belief 0 or 1) to the mediator, whenever the mediator commits to a reporting rule that recommends $a'$ 
	more often if the sender reports  $0$ and recommends $a''$ 
	more often if the sender reports $1$. Therefore, the sender's incentive constraints in the mediated communication game are satisfied if the mediator commits to the same information structure as the underlying experiment that induces our optimal random posterior.  
	In other words, interim incentive compatibility in our model is more stringent than the incentive compatibility restrictions required by the mediator model, and therefore our model delivers a (weakly) lower expected payoff for the sender than that achievable in \cite{Salamanca}.

	\newpage 
	
	\appendix
	
	\section*{Appendix}
	
	\begin{proof}[{\bf Proof of Lemma \ref{lem1}}]
		Given $P$, $\sigma_S$, and a message $m\in M$, the receiver forms a posterior belief $\hat{\mu}^m$, where
		\begin{equation*}
			\hat{\mu}^m(\theta)=\sum_{\mu \in \operatorname{supp}( P)}\frac{P(\mu) \sigma_S(m|\mu) }{\sum_{\mu\in \operatorname{supp}( P)} P(\mu) \sigma_S(m|\mu)} 	\mu(\theta),
		\end{equation*}
		for $\theta \in \Theta$.
		We use $\hat{P} \in \Delta(\Delta \Theta)$ to denote the distribution of the receiver's posterior beliefs, with 
		$\hat{P}(\hat{\mu}^m)=\sum_{\mu\in \operatorname{supp} (P) } P(\mu)\sigma_S(m|\mu)$.  Then player $i$'s expected utility can be simplified to:
		\begin{equation*}
			\label{belief2}
			U_i(\sigma_S,\sigma_R,P)
			=\sum_{m\in M, \, \theta\in \Theta,\,  a\in A} \hat{P}(\hat{\mu}^m) \hat{\mu}^m(\theta) \sigma_R(a|m) u_i(a,\theta).
		\end{equation*}
		Thus each player's expected utility only depends on the joint distribution of the receiver's posterior belief and the action. 
		If we let the sender directly commit to the random posterior $\hat{P}$, and construct a (truth-telling) reporting strategy $\hat\sigma_S$ such that for all $\mu \in \operatorname{supp} (\hat{P})$, $\hat\sigma_S (m|\hat{\mu}^m)=1$,
		then player $i$'s expected utility further simplifies to:
		\begin{equation*}
			U_i(\sigma_S,\sigma_R,P)=U_i(\hat\sigma_S,\sigma_R,\hat{P}).
		\end{equation*}
		Moreover, $(\sigma_S,\sigma_R,P)$ being an equilibrium strategy profile implies $(\hat\sigma_S,\sigma_R,\hat{P})$ is an equilibrium strategy profile. Since reporting $m\in M$ is a best response to $\sigma_R$ for every sender type of $\{\mu\in \operatorname{supp} (P): \sigma_S(m|\mu)>0\}$,
		reporting $m$ is also a best response for sender type $\hat{\mu}^m$, as $\hat{\mu}^m$ is a convex combination of $\{\mu\in \operatorname{supp} (P): \sigma_S(m|\mu)>0\}$.
		
		To prove the second part, suppose  a random posterior $P$ with $|\operatorname{supp} (P)|>|\Theta|$ that can 
		lead to a truth-telling equilibrium is optimal. By 
		Carath\'{e}odory's Theorem and Krein-Milman Theorem, 
		$\mu_0=\mathbb{E}_P[\mu]$ 
		can be written as a convex combination of 
		$|\Theta|$ elements of  $\operatorname{supp}( P)$, denoted as $P'\in \Delta (\Delta \Theta)$ with $|\operatorname{supp} (P')| =
		|\Theta|$ and 
		$\operatorname{supp} (P')\subset \operatorname{supp} (P)$.
		Let the receiver preserve $\sigma_R$, then $P'$ can 
		lead to a truth-telling equilibrium.  
		Let $c:\operatorname{co}\left(\operatorname{supp} (P)\right) 
		\rightarrow \mathbb{R}$  be the smallest concave function  such that $c(\mu)\ge  \sum_{\theta,a} \mu(\theta) \sigma_R(a|\mu) u_S(a,\theta)$ at all $\mu\in \operatorname{supp} (P)$. The
		random posterior $P'$ can perform equally well as $P$ because $c$
		must be affine on $\operatorname{co}\left(\operatorname{supp}( P)\right)$.
	\end{proof}

	\begin{proof}[{\bf Proof of Theorem \ref{theorem}}]

		\begin{claim}\label{c1}
			If the optimal random posterior has support $\{\mu_{-j},\mu_k\}$, and the receiver uses mixed strategies $\alpha_{-j}\in \Delta A_R(\mu_{-j})$ and $\alpha_{k}\in \Delta A_R(\mu_{k})$ (with full support) in the sender-preferred equilibrium, 
			then at least one of the following is true: 
			\begin{enumerate}
				\item[(i)]  $\alpha_{-j}=\alpha^{MM}_{-j}$ and $\alpha_{k}=\alpha^{MM}_{k}$; 
				\item[(ii)] $W^*(\mu_0)\in\{W^{PM}_{-j,k},W^{MP}_{-j,k},W^{PP}_{-j,k} \}$.
			\end{enumerate}
		\end{claim}
		
		\begin{proof}[Proof of Claim \ref{c1}]
			Since the receiver uses mixed strategies at each belief, 
			$\mu_{-j}\ne 0$ and $\mu_k\ne 1$. 
			Let $A_R(\mu_{-j})=\{a_{-j},a_{-j+1}\}$ and $A_R(\mu_{k})=\{a_{k-1},a_k\}$. 
			Since $\alpha_{-j}$ and $\alpha_{k}$ are incentive compatible for the sender, from ($\ref{IC-sender}$),
			\begin{equation}\label{mmapp}
				\mathbb{E}_{\alpha_{-j}}[u_S(a,\mu_{-j})]\ge \mathbb{E}_{\alpha_{k}}[u_S(a,\mu_{-j})],
				\qquad 
				\mathbb{E}_{\alpha_{k}}[u_S(a,\mu_{k})]\ge \mathbb{E}_{\alpha_{-j}}[u_S(a,\mu_{k})].
			\end{equation}
			Suppose $u_S(a_{-j},\mu_{-j})\ne  u_S(a_{-j+1},\mu_{-j})$ and $ u_S(a_{k-1},\mu_k) \ne u_S(a_{k},\mu_k)$.  
			Suppose further that both inequalities hold strictly and $\alpha_{-j}$ and $\alpha_{k}$ have full support. Then the sender can achieve a strictly higher expected payoff if the receiver deviates from $\alpha_{-j}$ by assigning a slightly larger probability on the sender-preferred action $\overline{a}_{-j}\in A_R(\mu_{-j})$. 
			As long as the increase in probability of choosing $\overline{a}_{-j}$ is small enough, such deviation
			would raise the payoff from truth-telling at belief $\mu_{-j}$ without violating the truth-telling constraint at belief $\mu_{k}$. 
			Suppose the first inequality in (\ref{mmapp}) holds as an equality and the second inequality in (\ref{mmapp}) holds strictly. Then the sender can achieve a strictly higher expected payoff under the same argument. If the first inequality in (\ref{mmapp}) holds strictly and the second inequality in (\ref{mmapp}) holds as an equality, a symmetric argument will apply. Therefore, the optimality of $\alpha_{-j}$ and $\alpha_{k}$ implies that both inequalities in (\ref{mmapp}) hold as an equality.

			If $m_S(a_{-j})$, $ m_S(a_{-j+1})$, $ m_S(a_{k-1})$ and $m_S(a_{k})$ are not all equal,
			then there is a unique $(\alpha^{MM}_{-j},\alpha^{MM}_k)\in \Delta A_R(\mu_{-j}) \times \Delta A_R(\mu_{k})$
			such that the constraints (\ref{mmapp}) hold as equalities.
			Next, if $m_S(a_{-j})= m_S(a_{-j+1})=m_S(a_{k-1})= m_S(a_{k})$, then there 
			exist infinitely many solutions. The optimality of $\alpha_{-j}$ and $\alpha_{k}$ will then imply that the receiver takes either $\overline{a}_{-j}$ at belief $\mu_{-j}$ and $ \alpha^{MM}_k$ at belief $\mu_k$, or $\overline{a}_{k}$ at $\mu_k$ and  $\alpha^{MM}_{-j}$ at belief $\mu_{-j}$. This 
			contradicts the premise that both $\alpha_{-j}$ and $\alpha_{k}$ have full support. Moreover, in this case $W^*(\mu_0)\in \{W^{PM}_{-j,k},W^{MP}_{-j,k}\}$.

			If at least one of $u_S(a_{-j},\mu_{-j})\ne  u_S(a_{-j+1},\mu_{-j})$ and $ u_S(a_{k-1},\mu_k) \ne u_S(a_{k},\mu_k)$ does not hold, we can slightly alter the above argument to show $W^*(\mu_0)\in \{W^{PM}_{-j,k},W^{MP}_{-j,k},W^{PP}_{-j,k} \}$. 
		\end{proof}
		
		\begin{claim}\label{c2}
			If the optimal random posterior has support $\{\mu_{-j},\mu_k\}$, and the receivers takes pure action $a'\in A_R(\mu_{-j})$ and mixed action $\alpha_k\in \Delta A_R(\mu_k)$ (with full support) in the sender-preferred equilibrium, then at least one of the following is true: 
			\begin{enumerate}
				\item[(i)]  $ a'=\overline{a}_{-j}$  and $\alpha_k=\alpha^{PM}_k$; 
				\item[(ii)] $W^*(\mu_0)\in\{W^{MM}_{-j,k},W^{MP}_{-j,k},W^{PP}_{-j,k}\}$.
			\end{enumerate}
		\end{claim}

		\begin{proof}[Proof of Claim \ref{c2}]
			Since the receiver uses mixed strategies at belief $\mu_k$, 
			$\mu_k\ne 1$ and $A_R(\mu_{k})=\{a_{k-1},a_k\}$. 
			Since $a'$ and $\alpha_{k}$ are incentive compatible for the sender, from ($\ref{IC-sender}$),
			\begin{equation}\label{pmapp}
				u_S(a',\mu_{-j})\ge \mathbb{E}_{\alpha_{k}}[u_S(a,\mu_{-j})], 
				\qquad      \mathbb{E}_{\alpha_{k}}[u_S(a,\mu_{k})]\ge u_S(a',\mu_{k}).
			\end{equation}
			
			Suppose that $ u_S(a_{k-1},\mu_k) \ne u_S(a_{k},\mu_k)$.
			We first show that the first inequality in (\ref{pmapp}) must hold as an equation.  Suppose to the contrary that this inequality holds strictly. Then the sender can achieve a strictly higher payoff if the receiver deviates from $\alpha_{k}$ by assigning a slightly larger probability on the sender-preferred action $\overline{a}_{k}\in A_R(\mu_{k})$. As long as the increase in probability of choosing $\overline{a}_k$ is small enough, such deviation would raise the sender's payoff from truth-telling at belief $\mu_{k}$ without violating (\ref{pmapp}) at belief $\mu_{-j}$, leading to a contradiction.

			Now, suppose the second inequality in (\ref{pmapp}) also holds as an equation.  Given the result established above, the sender is indifferent between $a'$ and $\alpha_k$ both at belief $\mu_{-j}$ and at belief $\mu_k$.  This case then reduces to the double-randomization case.  Therefore, $W^*(\mu_0)=W^{MM}_{-j,k}$, and part (ii) of this claim is satisfied.
			
			Next, suppose the second inequality in (\ref{pmapp}) holds strictly.  There are two possibilities: (1) the sender obtains different payoffs from $a_{-j}$ and $a_{-j+1}$ at belief $\mu_{-j}$, (2) the sender obtains the
			same payoff from $a_{-j}$ and $a_{-j+1}$ at belief $\mu_{-j}$.
			
			Case (1). Suppose $a'$ is not the sender-preferred action $ \overline{a}_{-j}$ at belief $\mu_{-j}$.
			Then the sender can achieve a strictly higher payoff 
			by inducing the receiver to deviate from $a'$ to a mixed strategy $\alpha_{-j}$ that assigns a small positive probability on $\overline{a}_{-j}$, without violating the incentive constraints.
			This shows that $a'$ must be equal to $\overline{a}_{-j}$.  Because $a'=\overline{a}_{-j}$, and the first condition of (\ref{pmapp}) as an equation implies that $\alpha_k=\alpha_k^{PM}$, and part (i) of this claim is satisfied.  
			
			Case (2).  When the sender obtains the same payoff from $a_{-j}$ and $a_{-j+1}$ at belief $\mu_{-j}$, the convention we adopt is $\overline{a}_{-j}=a_{-j+1}$.  Suppose $a' = a_{-j} \ne \overline{a}_{-j}$.  Then the optimality of $\{\mu_{-j},\mu_k\}$ implies that the sender is indifferent between $\alpha^{PM}_k$ and $a_{-j}$ at belief $\mu_k$. Otherwise, the random posterior with support $\{\mu_{-j-1},\mu_k\}$ can performs strictly better.  Therefore, the second inequality in (\ref{pmapp}) cannot hold strictly under the convention we adopt.  This contradiction implies that we must have $a'=a_{-j+1}=\overline{a}_{-j}$.  The first condition of (\ref{pmapp}) as an equation then implies that $\alpha_k=\alpha_k^{PM}$, and part (i) of this claim is satisfied.  
			
			Suppose $ u_S(a_{k-1},\mu_k)= u_S(a_{k},\mu_k)$. If the first inequality in (\ref{pmapp}) holds with equality, then we can use the same argument as above. Otherwise, if the first inequality in (\ref{pmapp}) holds strictly, then we can slightly alter the above argument to show that $W^*(\mu_0)\in\{W^{MM}_{-j,k},W^{MP}_{-j,k},W^{PP}_{-j,k}\}$. 
		\end{proof}

		\begin{claim}
			\label{c3}
			If the optimal random posterior has a support $\{\mu_{-j},\mu_k\}$, and the receiver uses only pure strategies $a'\in A_R(\mu_{-j})$ and $a''\in A_R(\mu_{k})$ in the sender-preferred equilibrium, then either one of the following is true: 
			\begin{enumerate}
				\item[(i)]  $a'=\overline{a}_{-j}$  and $a''= \overline{a}_{k}$; 
				\item[(ii)] $W^*(\mu_0)\in\{W^{PM}_{-j,k},W^{MP}_{-j,k},W^{MM}_{-j,k}\}$.
			\end{enumerate} 
		\end{claim}
		
		\begin{proof}[Proof of Claim \ref{c3}]
			If $\mu_{-j}=0$ and $\mu_{k}=1$, then $A_R(\mu_{-j})=\overline{a}_{-j}$, $A_R(\mu_{k})=\overline{a}_{k}$, and part (i) is satisfied.
			
			If $\mu_{-j}\ne 0$ and $\mu_{k}\ne 1$, then $A_R(\mu_{-j})=\{a_{-j},a_{-j+1}\}$, $A_R(\mu_{k})=\{a_{k-1},a_k\}$, and there are four possibilities. (1) The sender obtains different payoffs from $a_{-j} $ and $ a_{-j+1}$ at belief $\mu_{-j}$, and different payoffs from  $a_{k-1}$ and $a_k$ at belief $\mu_k$. (2) The sender obtains same payoff from $a_{-j} $ and $ a_{-j+1}$ at belief $\mu_{-j}$, but different payoffs from  $a_{k-1}$ and $a_k$ at belief $\mu_k$. (3)  The sender obtains different payoffs from $a_{-j} $ and $ a_{-j+1}$ at belief $\mu_{-j}$, and same payoff from  $a_{k-1}$ and $a_k$ at belief $\mu_k$. (4)  The sender obtains same payoff from $a_{-j} $ and $ a_{-j+1}$ at belief $\mu_{-j}$, and same payoff from  $a_{k-1}$ and $a_k$ at belief $\mu_k$.
			
			Case (1). By incentive compatibility,
			\begin{equation}\label{ppapp}
				u_S(a',\mu_{-j})\ge  u_S(a'',\mu_{-j}), 
				\qquad
				u_S(a'',\mu_{k})\ge u_S(a',\mu_{k}).
			\end{equation}
			Suppose both inequalities hold strictly. 
			The optimality of $(a',a'')$ implies $a'= \overline{a}_{-j}$ and $a''= \overline{a}_{k}$, with a reasoning similar to that
			in the proof of Claim \ref{c1}. Thus, part (i) is satisfied. 
			Suppose the first inequality holds as equality and 
			the second inequality holds strictly. Then the optimality of $(a',a'')$  implies  $a'= \overline{a}_{-j}$, with a 
			reasoning similar to that in the proof of Claim \ref{c2}. Moreover, 
			when $a'=\overline{a}_j$ and the first inequality holds as equality,
			we have $a''= \alpha^{PM}_k$.  Therefore, $W^*(\mu_0)=W^{PM}_{-j,k}$, and part (ii) is satisfied. 
			Similarly, if the first inequality hold strictly and  the second inequality holds as an equality, then  $W^*(\mu_0)=W^{MP}_{-j,k}$ and part (ii) is also satisfied. 
			Finally, suppose that		
			both inequalities hold as an equality. Then if $m_S(a_{-j})$, $ m_S(a_{-j+1})$, $ m_S(a_{k-1})$ and $m_S(a_{k})$ are not all equal, $a'=\alpha^{MM}_{-j}$ and $a''=\alpha^{MM}_{k}$. Therefore, $W^*(\mu_0)=W^{MM}_{-j,k}$. If $m_S(a_{-j})$, $ m_S(a_{-j+1})$, $ m_S(a_{k-1})$ and $m_S(a_{k})$ are equal, then $W^*(\mu_0)\in \{W^{PM}_{-j,k},W^{MP}_{-j,k}\}$. Part (ii) is again satisfied. 
			
			Case (2a).  If $a'=a_{-j}\ne \overline{a}_{-j}$, then 
			the optimality of $\{\mu_{-j},\mu_k\}$ implies that the sender is indifferent between $a''$ and $a_{-j}$ at belief $\mu_k$;  otherwise the random posterior with support $\{\mu_{-j-1},\mu_k\}$ is strictly better. Moreover, if the sender is indifferent between $a_{-j}$ and $a''$ at belief $\mu_{-j}$, then $W^*(\mu_0)=W^{MM}_{-j,k}$. On the other hand,  if the sender strictly prefers $a_{-j}$ over $a''$ at belief $\mu_{-j}$, then the optimality of $a''$ implies $a''=\overline{a}_k$. That is, $W^*(\mu_0)=W^{MP}_{-j,k}$.  In both sub-cases, part (ii) is satisfied. 
			Case (2b). If $a'\ne a_{-j}$, then $a'= \overline{a}_{-j}$.
			Then if the first inequality in (\ref{ppapp}) hold strictly, the optimality of $(\overline{a}_{-j},a'')$ implies that $a''=\overline{a}_k$, and part (i) is satisfied. On the other hand, if the first inequality in (\ref{ppapp}) hold as equality, then  $W^*(\mu_0)=W^{PM}_{-j,k}$, and part (ii) is satisfied. 
			
			Case (3) is symmetric
			to case (2), and if
			we apply all arguments above, case (4) implies $W^*(\mu_0)\in \{W^{PP}_{-j,k},W^{PM}_{-j,k},W^{MP}_{-j,k},W^{MM}_{-j,k}\}$.
			
			Finally, if either $\mu_{-j}=0$ or $\mu_{k}=1$, then with a similar reasoning we can conclude $W^*(\mu_0)\in \{W^{PP}_{-j,k},W^{PM}_{-j,k},W^{MP}_{-j,k},W^{MM}_{-j,k}\}$.
		\end{proof}
		
		Claims \ref{c1}--\ref{c3} (together with an analogous claim for the case of MP)  imply that it is sufficient to only focus on $W^{PP}_{-j,k}$, $W^{PM}_{-j,k}$, $W^{MP}_{-j,k}$, and $W^{MM}_{-j,k}$. 
		
		Lastly, in step 2 of the algorithm, we only consider $(-j,k)\in S_1$. To see the sufficiency of it, suppose a pair of $(-j,k)$ outside $S_1$ is IC-PP, or IC-PM, or IC-MP, or IC-MM. Then,
		\begin{equation*}
			\max \{  W^{PP}_{-j,k}, W^{PM}_{-j,k}, W^{MP}_{-j,k} , W^{MM}_{-j,k} \} 
			\le W_{-j,k}(\mu_0;\overline{a}_{-j},\overline{a}_k)
			\le W^1. 
		\end{equation*}
		The first inequality follows from the fact that
		the sender's value correspondence $v(\mu_{-j})$ is lower than $u_S(\overline{a}_{-j},\mu_{-j})$, and  $v(\mu_{k})$ is lower than $u_S(\overline{a}_{k},\mu_{k})$. The second inequality 
		comes
		from the construction of $S_1$.  So even if
		the random posterior with support $\{\mu_{-j},\mu_k\}$ outside $S_1$ is incentive compatible, it will never improve on the outcome from step 1. Moreover, in step 3
		of the algorithm, we only consider $(-j,k)\in S_2$. To see the sufficiency of it, suppose a pair of $(-j,k)$ outside $S_2$ is IC-MM. Then,
		\begin{equation*}
			W_{-j,k}^{MM} \le \max\{ W_{-j,k}(\mu_0;\overline{a}_{-j},\alpha^{PM}_k) , W_{-j,k}(\mu_0;\alpha^{PM}_{-j},\overline{a}_{k}) \}
			\le \max\{ W^{(a)} , W^{(b)} \} = W^2.
		\end{equation*}
		The second inequality is from our construction of $S_2$.
		To see the first inequality, notice that under IC-MM, the sender is indifferent between $\alpha^{MM}_{-j} $ and $\alpha^{MM}_{k}:=(\gamma_{k},1-\gamma_{k})$ at belief $\mu_{-j}$,
		\begin{equation*}
			\mathbb{E}_{\alpha_{-j}}[u_S(a,\mu_{-j})]=\gamma_k u_S(\overline{a}_k,\mu_{-j})+(1-\gamma_k)u_S(\underline{a}_{k},\mu_{-j}).
		\end{equation*}
		From our construction of $\alpha^{PM}_k:=(\gamma'_k,1-\gamma'_k)$, 
		\begin{equation*}
			u_S(\overline{a}_{-j},\mu_{-j})=\gamma'_k u_S(\overline{a}_k,\mu_{-j})+(1-\gamma'_k)u_S(\underline{a}_{k},\mu_{-j}).
		\end{equation*}
		Therefore $\mathbb{E}_{\alpha_{-j}}[u_S(a,\mu_{-j})]\le u_S(\overline{a}_{-j},\mu_{-j})$ implies 
		$\gamma_k\le \gamma'_k$, which further implies that $\mathbb{E}_{\alpha^{MM}_{k}}[u_S(a,\mu_{k})]\le \mathbb{E}_{\alpha^{PM}_{k}}[u_S(a,\mu_{k})]$. 
		We thus have $W_{-j,k}^{MM} \le   W_{-j,k}(\mu_0;\overline{a}_{-j},\alpha^{PM}_k) $. For a similar reason, $W_{-j,k}^{MM} \le   W_{-j,k}(\mu_0;\alpha^{PM}_{-j},\overline{a}_{k}) $. Notice that 
		this argument does not require
		$\{\mu_{-j},\mu_k\}$ 
		to be IC-PM or IC-MP.
	\end{proof}

	\begin{proof}[{\bf Proof of Proposition \ref{dominate}, part (b)}]
		
		Let $a_n$ be the least-preferred action for the sender in state 0 and $a_l$ be his least-preferred action in state 1.  (If there are multiple least-preferred actions in one state, just pick any one of them.)	We have $a_n \ne a_l$, otherwise $a_n$ would be a worst action. Moreover, $u_S(a_l,0)>u_S(a_n,0)$ and $u_S(a_n,1)>u_S(a_l,1)$.
		This implies $m_S(a_l)<m_S(a_n)$. By aligned marginal incentives, $m_R(a_l) < m_R(a_n)$, and therefore the interval of beliefs for which $a_l$ is receiver's best response, denoted $I_l$, is to the left of the interval $I_n$ for $a_n$.  Following the same  convention adopted in the text, 
		we let $\underline{I}_{l}$ represent the lowest belief in $I_l$ and let $\overline{I}_{n}$ represent the highest belief in $I_n$.
		
		There are three mutually exclusive cases. (1) $a_l$ and $a_n$ are strictly IC-PP for 
		$\{\underline{I}_{l},\overline{I}_{n}\}$; (2) $a_l$ blocks $a_n$ (but $a_n$ does not block $a_l$); and (3) $a_n$ blocks $a_l$ (but $a_l$ does not block $a_n$).  In case (1), information design is valuable, for example
		when $\mu_0$ is in the interior of $I_l$.  
		Cases (2) and (3) are symmetric; thus we consider case (2) only. 
		
		Denote $a_{n+1}$ as the next action higher than $a_n$; i.e., the receiver is indifferent between $a_{n+1}$ and $a_n$ at belief $\overline{I}_n$.  There are several possibilities:
		
		\begin{enumerate}
			\item[(2a)] Suppose $a_l$ is (weakly) worse than $a_{n+1}$ at  belief  $\underline{I}_{l}$. Note that under case (2) $a_l$ is better than  $a_n$ at both belief $\underline{I}_{l}$ and $ \overline{I}_n$. Therefore, there is a mixed action $\alpha_n \in \Delta \{a_n,a_{n+1}\}$ such that 
			the sender with belief $\underline{I}_l$
			is indifferent between $a_l$ and $\alpha_n$. Moreover, by aligned marginal incentives, $m_S(\alpha_n)>m_S(a_l)$. Thus, the random posterior with support $\{\underline{I}_{l},\overline{I}_{n}\}$ is IC-PM for action $a_{l}$ and some mixed action $\alpha_n$ and information design is valuable, for example when $\mu_0$ is in the interior of $I_l$. 
			
			\item[(2b)] Suppose $a_l$ is (strictly) better than $a_{n+1}$ at  belief  $\underline{I}_{l}$.  
			\begin{enumerate}
				\item[(i)] If $a_{n+1}$  is better than $a_l$  at belief  $\overline{I}_{n+1}$, then $\{\underline{I}_{l},\overline{I}_{n+1}\}$ is IC-PP. 
				\item[(ii)] If $a_{n+1}$ is worse than $a_l$  at belief  $\overline{I}_{n+1}$, then let $a_{n'}$ be the highest action that $a_l$ blocks.  Notice that $a_{n'}<a_K$, where $a_K = A_R(1)$, because $m_S(a_K)>m_S(a_{n})$ and $u_S(a_K,0)\ge u_S(a_n,0)$ imply that $u_S(a_K,1)>u_S(a_n,1)$. 
				Since $a_l$ is the least-preferred action in state 1, we have $u_S(a_n,1) \ge u_S(a_l,1)$, and thereby $u_S(a_K,1) > u_S(a_l,1)$.  Therefore $a_l$ does not block $a_K$.  
				Then with a similar argument as in (2a) and (2b-i), either of the following is true: $\{\underline{I}_{l},\overline{I}_{n'}\}$ is IC-PM, or $\{\underline{I}_{l},\overline{I}_{n'+1}\}$ is IC-PP. The existence of $a_{n'+1}$ comes from $a_{n'}<a_K$.\hfil\qedhere 
			\end{enumerate}
		\end{enumerate}
	\end{proof}

 \begin{proof}[{\bf Proof of Proposition \ref{state-inde}}]
		The ``only if'' part is simple.  If the sender's ranking is monotone in the index of the actions, then there does not exist an informative equilibrium outcome in which the receiver chooses different actions (including mixed actions) after different messages. This implies that information design is not valuable. 
		
		To show the ``if'' part, suppose the sender's ranking is non-monotone in the index of the actions. This implies that there must be at least three actions in $ A$.
		Moreover there exists an index $n$ such that either (1) the sender prefers $a_{n-1}$ to $a_{n}$, but $a_{n+1}$ is  ranked above $a_n$; or (2) sender prefers $a_n$ to $a_{n-1}$, but $a_{n+1}$ is  ranked below $a_n$.  Let $\underline{I}_n:=\min \{\mu: a_n\in A_R(\mu)\}$ and $\overline{I}_n:=\max \{\mu: a_n\in A_R(\mu)\}$. 
		In case (1a), the sender prefers $a_{n+1}$ to $a_{n-1}$ to $a_n$ at all beliefs, including at belief $\underline{I}_{n-1}$.  Therefore, there exists a mixture $\alpha_n\in \Delta \{a_n,a_{n+1}\}$ 
		that the receiver would optimally choose at belief 
		$\overline{I}_n$ such that the sender is indifferent between $a_{n-1}$ and $\alpha_n$
		at belief $\underline{I}_{n-1}$.  Moreover, because $m_S(\alpha_n) > m_S(a_{n-1})$, the random posterior with support $\{\underline{I}_{n-1},\overline{I}_n\}$ and an expectation $\mu_0\in (\underline{I}_{n-1},\overline{I}_n)$ is IC-PM given the receiver optimally chooses between $a_{n-1}$ and 
		$\alpha_n$.  
		In case (1b), the sender prefers $a_{n-1}$ to $a_{n+1}$ to $a_n$ at any belief. With a similar reasoning, the random posterior with support $\{\underline{I}_{n},\overline{I}_{n+1}\}$  is IC-MP given the receiver optimally chooses between some $\alpha_{n-1}\in \Delta \{a_{n-1},a_n\}$ and $a_{n+1}$.  
		In case (2a), the sender prefers $a_n$ to $a_{n-1}$ to $a_{n+1}$. 
		Then the random posterior with support $\{\underline{I}_{n-1},\overline{I}_n\}$  is IC-PM given the receiver optimally chooses between $a_{n-1}$ some $\alpha'_{n}\in \Delta \{a_{n},a_{n+1}\}$.
		In case (2b), the sender prefers $a_n$ to $a_{n+1}$ to $a_{n-1}$. Then the random posterior with support $\{\underline{I}_{n},\overline{I}_{n+1}\}$   is IC-MP given the receiver optimally chooses between some $\alpha'_{n-1}\in \Delta \{a_{n-1},a_n\} $ and $a_{n+1}$. 
	\end{proof}

	\begin{proof}[{\bf Proof of Proposition \ref{compwithkg}}]
		With aligned marginal incentives, the concavification 
		result in \cite{kamenica2011bayesian}
		implies that there exists 
		an $n'<n$ such that with prior belief $\mu_0\in(\underline{I}_{n'},\overline{I}_{n'})$, the optimal experiment under full commitment
		has support $\{\underline{I}_{n'},\underline{I}_{n}\}$. Similarly, there exists an $n''>n$ such that
		with a different prior belief $\mu'_0\in (\underline{I}_{n''},\overline{I}_{n''})$, the optimal 
		experiment under full commitment
		has support $\{\overline{I}_{n},\overline{I}_{n''}\}$. We want to show that, in our model, the optimal experiment is strictly more informative than that under full commitment either when the prior is $\mu_0$ or when the prior is $\mu'_0$.
		There are only two cases. 
		(a) There exists a $k\ge n$ such that $u_S(a_{k},\underline{I}_{n'})\ge u_S(a_{n'},\underline{I}_{n'})\ge u_S(a_{k+1},\underline{I}_{n'})$. (b) For every $k\ge n$,  $u_S(a_{k},\underline{I}_{n'})> u_S(a_{n'},\underline{I}_{n'})$.
		Notice that if there exists an experiment with support $\{\underline{I}_{n'},\mu\}$ where $\mu\in (\overline{I}_{n'},\underline{I}_{n}]$ that is incentive compatible, then the sender's expected payoff from such experiment is smaller than $u_S(a_{n'},\underline{I}_{n'})+m_S(a_n) (\mu_0-\underline{I}_{n'})$ (this is implied by Lemma \ref{l:IC}). 
		Similarly, if there exists an experiment with support $\{\mu,\overline{I}_{n''}\}$ where $\mu\in (\overline{I}_{n},\overline{I}_{n''})$ that is incentive compatible, then the sender's expected payoff under $\mu'_0$ from such experiment is smaller than $u_S(a_{n''},\overline{I}_{n''})-m_S(a_n) (\overline{I}_{n''}-\mu'_0)$.
		
		In case (a), the experiment with support $\{\underline{I}_{n'},\overline{I}_{k}\}$ is IC-PM given aligned marginal incentives. Moreover, sender's expected payoff from such experiment is greater than   $u_S(a_{n'},\underline{I}_{n'})+m_S(a_n) (\mu_0-\underline{I}_{n'})$. Because the receiver randomizes between $a_k$ and $a_{k+1}$ at belief $\overline{I}_{k}$ and thereby the marginal incentive from such randomization $\alpha_k$ is greater than $m_S(a_n)$. Also, from the construction of IC-PM, the sender's expected payoff equals $u_S(a_{n'},\underline{I}_{n'})+m_S(\alpha_k) (\mu_0-\underline{I}_{n'})$. Therefore, under $\mu_0$, there exists an experiment with support $\{\underline{I}_{n'},\overline{I}_{k}\}$ that is better than any IC experiment with support $\{\underline{I}_{n'},\mu\}$ where $\mu\in (\overline{I}_{n'},\underline{I}_{n}]$.

		Moreover, there cannot exist an IC experiment $\{\mu,\mu'\}$ with $\mu<\underline{I}_{n'}$ and $\overline{I}_{n'}<\mu'<\underline{I}_{n}$ that 
		yields the sender an expected payoff higher than $u_S(a_{n'},\underline{I}_{n'})+m_S(a_n) (\mu_0-\underline{I}_{n'})$. To see 
		this point, note that Lemma \ref{l:IC} implies that 
		the slope of sender's expected payoff is smaller than $m_S(a')$ where $a'\in A_R(\mu')$, which 
		in turn 
		is smaller than $m_S(a_n)$ from aligned marginal incentives. This implies the sender's expected payoff from such an experiment, if the prior belief is $\underline{I}_{n'}$, would be higher than $u_S(a_{n'},\underline{I}_{n'})$, which 
		contradicts the fact that $u_S(a_{n'},\underline{I}_{n'})$ lies on the concave envelope of $\overline{v}(\cdot)$.
		
		Since $\overline{I}_k > \underline{I}_n$, the optimal experiment in our model under $\mu_0$, which has support $\{\underline{I}_{n'},\overline{I}_k\}$, is 
		strictly more informative than the 
		(full commitment) experiment with support
		$\{\underline{I}_{n'},\underline{I}_{n}\}$.
		
		In case (b), since $u_S(a_{n''},\underline{I}_{n'})>u_S(a_{n'},\underline{I}_{n'})$, we have $u_S(a_{n''},\overline{I}_{n''})>u_S(a_{n'},\overline{I}_{n''}) $ under the assumption of aligned marginal incentive. Then there must exist 
		an $a_k$ with $n' < k \le n$
		such that $u_S(a_k,\overline{I}_{n''})\ge u_S(a_{n''},\overline{I}_{n''})\ge u_S(a_{k-1},\overline{I}_{n''})$. Therefore, there exists an IC-MP experiment with support $\{\underline{I}_{k},\overline{I}_{n''}\}$ such that the receiver randomizes between $a_{k}$ and $a_{k-1}$ at belief $\underline{I}_{k}$ and such randomization $\alpha_k$ has a marginal incentive 
		$m_S(\alpha_k)$ smaller than $m_S(a_n)$. 
		Under $\mu'_0$, this experiment generates the sender an expected payoff higher than $u_S(a_{n''},\overline{I}_{n''})-m_S(a_n) (\overline{I}_{n''}-\mu'_0)$. Moreover, such experiment is better than any IC experiment with support  $\{\mu,\overline{I}_{n''}\}$ where $\mu\in [\overline{I}_{n},\overline{I}_{n''})$. 
		Because $\underline{I}_k < \underline{I}_n$, the optimal experiment in our model under $\mu'_0$, which has support $\{ \underline{I}_k,\overline{I}_{n''}\}$, is strictly more informative than the (full commitment) experiment with support 
		$\{\overline{I}_{n},\overline{I}_{n''}\}$.
	\end{proof}
	
	\newpage
	\bibliographystyle{aer}
	\bibliography{Bayesian.bib}
\end{document}